\documentclass{article}
\usepackage{bm}
\usepackage[font=small]{caption}
\RequirePackage{amsthm,amsmath,amsfonts,amssymb}
\usepackage{parskip}
\usepackage[a4paper,left=3cm,right=2cm,top=2.5cm,bottom=2.5cm]{geometry}

\usepackage{bm}
\usepackage{natbib}
\usepackage[ruled]{algorithm2e}
\usepackage[dvipsnames]{xcolor}
\usepackage[colorlinks = true,
            linkcolor = NavyBlue,
            urlcolor  = NavyBlue,
            citecolor = NavyBlue,
            anchorcolor = NavyBlue]{hyperref} 

\usepackage[graphicx]{realboxes}
\usepackage{enumerate}
\usepackage{comment}
\RequirePackage{subfigure}
\RequirePackage{algorithmic}

\author{}

\RequirePackage{amsthm}

\theoremstyle{plain}
\newtheorem{remark}{Remark}
\newtheorem{theorem}{Theorem}
\newtheorem{definition}{Definition}
\newtheorem{proposition}{Proposition}

\numberwithin{equation}{section}

\newcommand\ma{\mathbf{A}}
\newcommand\mb{\mathbf{B}}
\newcommand\mc{\mathbf{C}}
\newcommand\md{\mathbf{D}}
\newcommand\me{\mathbf{E}}
\newcommand\mh{\mathbf{H}}
\newcommand\mi{\mathbf{I}}

\newcommand\mr{\mathbf{R}}

\newcommand\mj{\mathbf{J}}
\newcommand\mx{\mathbf{X}}

\newcommand\mw{\mathbf{W}}

\newcommand\mm{\mathbf{M}}

\newcommand\mpp{\mathbf{P}}
\newcommand\mo{\mathbf{O}}

\newcommand\muu{\mathbf{U}}

\newcommand\mLambda{\bm{\Lambda}}
\newcommand*{\supp}{\mathrm{supp}}

\usepackage{etoolbox}
\makeatletter
\csundef{listofalgorithms}
\makeatother

\title{Dynamic networks clustering via mirror distance}
\date{}

\author{Runbing Zheng\thanks{Department of Applied Mathematics and Statistics, Johns Hopkins University}
\and Avanti Athreya\footnotemark[1]
\and Marta Zlatic\thanks{MRC Laboratory of Molecular Biology, University of Cambridge}
\and Michael Clayton\footnotemark[3]
\and Carey E. Priebe\footnotemark[1]}
\begin{document}

\maketitle

\renewcommand{\thefootnote}{}
\footnotetext{This work was supported in part by Office of Naval Research (ONR) Science of Autonomy Award Number N00014-24-1-2278.}
\renewcommand{\thefootnote}{\arabic{footnote}}

\begin{abstract}
The classification of different patterns of network evolution, for example in brain connectomes or social networks, is a key problem in network inference and modern data science.   Building on the notion of a network's Euclidean mirror, which captures its evolution as a curve in Euclidean space, we develop the Dynamic Network Clustering through Mirror Distance (DNCMD), an algorithm for clustering dynamic networks based on a distance measure between their associated mirrors.
	We provide theoretical guarantees for DNCMD to achieve exact recovery of distinct evolutionary patterns for latent position random networks both when underlying vertex features change deterministically and when they follow a stochastic process.
	We validate our theoretical results through numerical simulations and demonstrate the application of DNCMD to understand edge functions in Drosophila larval connectome data, as well as to analyze temporal patterns in dynamic trade networks.	
\end{abstract}

\noindent%
{\it Keywords:} multiple dynamic networks, cluster analysis, spectral embedding

\section{Introduction}

%
%
%
%
 A network is a powerful way to represent the relationships or connections among a collection of objects, and a dynamic network can further describe how the structure of these relationships or connections evolves over time, with important applications in many fields, including social network analysis \cite{zuzul2021dynamic,wasserman1994social,kumar2006structure,palla2007quantifying}, neuroscience \cite{eschbach2020recurrent,hutchison2013dynamic,fox2005human}, and finance \cite{ellington2020dynamic,karim2022determining,wanke2019dynamic}.
For example, \cite{wang2021optimal,chen2024euclidean} study changepoint detection in a dynamic network whose evolution is governed by an underlying random walk process; \cite{yu2018netwalk, aggarwal2011outlier, manzoor2016fast} explore anomaly detection in single dynamic networks;  \cite{aynaud2010static, aynaud2013communities, xu2014dynamic} investigate community discovery in dynamic networks; and in \cite{pensky2024clustering, wang2023multilayer}, the authors focus on clustering and online changepoint detection in diverse multilayer networks. 
In this paper, we also study dynamic networks, and explore the relationships among multiple dynamic networks, aiming to cluster them based on their evolutionary patterns.

Brain neural networks provide a motivating example for the dynamic networks clustering problem: 
a practical problem in neuroscience is to identify synaptic connections that have significant influence during specific neural processes \cite{ko2011functional,song2005highly}.
For instance, during processes such as learning associations between stimuli and rewards or punishments, the evolution of brain neural networks can be studied through real experiments or simulations \cite{eschbach2020recurrent}.
In \cite{eschbach2020recurrent}, neuroscientists exploit techniques to simulate the time-varying behavior of brain neural networks by the removal of specific edges between specific pairs of neurons. For each edge, replicate simulations are conducted, resulting in multiple dynamic networks corresponding to the initial removal of a specific edge.
We wish to analyze how dynamic networks evolve, and in particular whether the associated removed edge significantly impacts subsequent network development. 
If clustering reveals that the dynamic networks associated with a particular removed edge are highly similar to each other but significantly different from those of other edges, it suggests that this edge is important and has a potentially unique function in this process. 
Discerning patterns in network evolution is not restricted to neural networks. Related important inference questions arise in various domains, including the evolution of distinct subcommunities in social networks \cite{palla2007quantifying}, interaction patterns in organizational networks  \cite{zuzul2021dynamic}, temporal changes in transportation networks \cite{agterberg2022joint}, and anomalies in trade networks \cite{wang2023multilayer, de2014network,chaney2014network}. In all of these  scenarios, a principled methodology for clustering networks based on evolutionary patterns would help address key application questions of interest.

Due to the high dimensionality and sparsity common in network data, traditional clustering methods like k-means \cite{ahmed2020k} are typically ineffective for accurately clustering dynamic networks by their distinct patterns of evolution. 
To address this challenge, we develop a dynamic network clustering method based on the Euclidean \textit{mirrors} proposed in \cite{athreya2024discovering}.
These mirrors constructed for individual dynamic networks are constructed using spectral decompositions of observed networks and provide a low-dimensional representation of network evolution.
Spectral embedding methods are widely used and have proven effective for multi-sample network inference; see \cite{levin2017central,jones2020multilayer,arroyo2021inference,gallagher2021spectral,jing2021community,pantazis2022importance,agterberg2022joint} for recent advancements.
In \cite{athreya2024discovering}, the authors assume that each node in the network has an associated time-varying low-dimensional latent vector of feature data, referred to as a latent position process, with the probabilities of the connections between any two nodes determined by their corresponding vectors, and the time-varying evolution of the latent vectors exhibiting a low-dimensional manifold structure.
For a dynamic network, the differences between the latent position matrices at observed time points are measured pairwise, and then classical multidimensional scaling (CMDS) \cite{borg2007modern,wickelmaier2003introduction,li2020central} is applied to obtain a configuration in a low-dimensional space that approximately preserves the dissimilarity. This results in a curve in low-dimensional Euclidean space that mirrors the evolution of the network dynamics over time, and is called the mirror of the dynamic network.
Given the observed adjacency matrices of the network time series, the latent position matrices can be estimated using spectral embeddings \cite{tang2018limit,sussman2012consistent}, with CMDS then providing an estimate of the mirror. 
The mirror method captures important features of the network evolution and has been demonstrated to be effective for visualizing a dynamic network and conducting inference tasks related to dynamic networks, including changepoint detection \cite{athreya2024discovering,chen2023discovering,chen2024euclidean}. 

For the problem of clustering evolutionary patterns of multiple dynamic networks, 
we consider the case where the differences in the evolution of the clusters can be reflected by the mirrors, and propose the Dynamic Network Clustering through Mirror Distance (DNCMD).
The configuration obtained by CMDS is centered at the origin but has an inherent issue of identifiability concerning orthogonal transformations.
Therefore, we construct a measure of the differences between estimated mirrors by solving an orthogonal Procrustes problem \cite{schonemann1966generalized}.
After obtaining the distances between the estimated mirrors of all pairs of dynamic networks, hierarchical clustering \cite{nielsen2016hierarchical,ward1963hierarchical} is applied to produce a dendrogram and the clustering of these associated mirrors based on the computed distances. 

DNCMD has been validated as effective for real-world complex problems, such as the neuroscience motivating example above. As shown in Figure~\ref{fig:intro}, we analyze the dynamic evolution of neural networks after removing three different synaptic edge connections (see details in Section~\ref{sec:brain}), and DNCMD successfully clusters most neural networks replicates for different removed connections and their impact on neural activity. This demonstrates its potential for identifying key synaptic connections, which we illustrate in Section~\ref{sec:brain}.
\begin{figure}[htbp] 
\centering
\subfigure{%
\includegraphics[width=10cm]{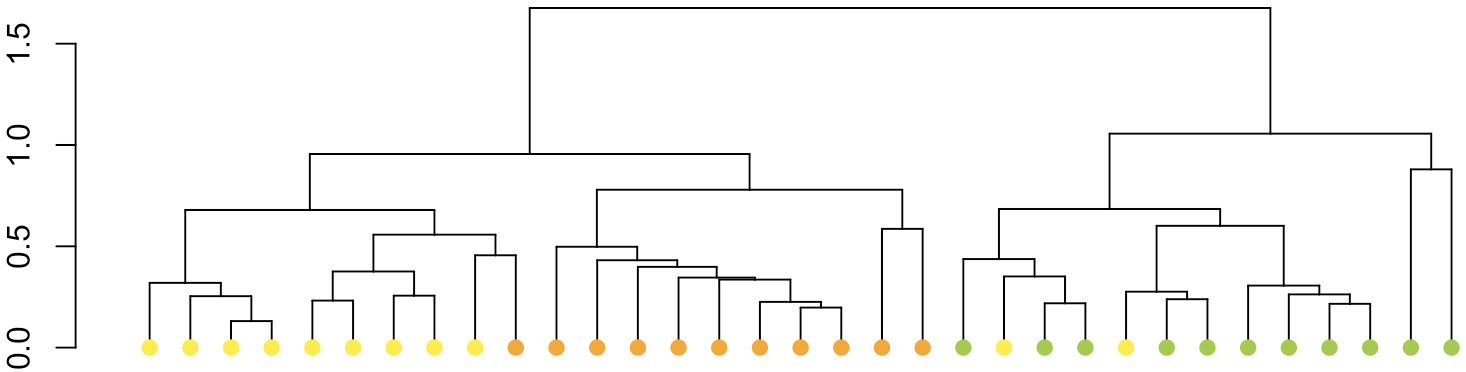}}
\subfigure{%
\includegraphics[width=4.5cm]{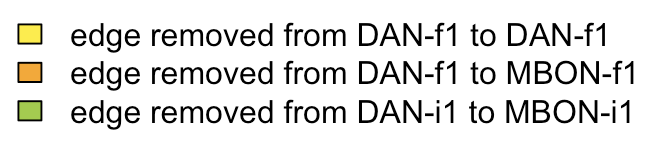}}
\caption{Dendrogram of dynamic networks of neurons for the removal of each single edge, obtained using the DNCMD algorithm. For each removed edge we have $11$ replicates.}
\label{fig:intro}
\end{figure}

For theoretical analysis, for a dynamic network, the evolution of the latent position matrix can be considered deterministic, as is the case when the features of the nodes in a network follow some predictable time-dependent pattern; it can also be random, as is the case when the actors in a network have underlying preferences that are subject to random shocks.
\cite{athreya2024discovering} considers the latent position matrix to be random and assumes that all its rows, representing the latent positions of individual nodes, are generated from a stochastic process, and proves that the mirror of the dynamic network can be consistently estimated.
In this paper, we prove the consistency result for the estimate of the mirror in the case where the latent position matrix is deterministic, allowing for flexibility in its evolution. Furthermore, for both cases, deterministic and random latent positions, we provide theoretical guarantees for DNCMD. We prove that DNCMD achieves exact recovery under mild conditions, where the term exact recovery means the clusters are correctly recovered with high probability.

The structure of our paper is as follows. In Section~\ref{sec:model}, we introduce the model for multiple dynamic networks, where the clustering structure can be distinguished by the mirror, and in this section, we first consider the case where the latent position matrices are deterministic. 
In Section~\ref{sec:alg}, we propose DNCMD to obtain the clustering result from the adjacency matrices at observed time points for multiple dynamic networks. 
In Section~\ref{sec:thm}, we first provide the theoretical guarantee for DNCMD for deterministic latent positions, and then in Section~\ref{sec:random X}, we briefly discuss the case where random latent positions are generated from stochastic processes, as considered in \cite{athreya2024discovering}, and we build the model for multiple dynamic networks with a clustering structure and provide the theoretical guarantee for DNCMD in this scenario as well.
Numerical simulations to demonstrate our theoretical results and show the clustering performance of DNCMD are presented in Section~\ref{sec:simu}. 
Section~\ref{sec:real} presents real data experiments to analyze synaptic functions in brain neural networks and temporal patterns in trade dynamic networks.
Detailed proofs of stated results and additional experiment results are provided in the supplementary material.

\subsection{Notations}

We summarize some notations used in this paper. 
For any positive integer $n$, we denote by $[n]$ the set $\{1,2,\dots, n\}$. 
For two non-negative sequences
$\{a_n\}_{n \geq 1}$ and $\{b_n\}_{n \geq 1}$, we write $a_n \lesssim
b_n$ ($a_n \gtrsim b_n$, resp.) if there exists some constant $C>0$
such that $a_n \leq C b_n$ ($a_n \geq C b_n$, resp.) for all  $n \geq 1$, and we write $a_n \asymp b_n$ if $a_n\lesssim b_n$ and $a_n\gtrsim b_n$.
The notation $a_n \ll b_n$ ($a_n \gg b_n$, resp.) means that there exists some sufficiently small (large, resp.) constant $C>0$ such that $a_n \leq Cb_n$ ($a_n \geq Cb_n$, resp.).
If $a_n/b_n$ stays bounded away from $+\infty$, we write $a_n=O(b_n)$ and $b_n=\Omega(a_n)$, and we use the notation $a_n=\Theta(b_n)$ to indicate that $a_n=O(b_n)$ and $a_n=\Omega(b_n)$.
If $a_n/b_n\to 0$, we write $a_n=o(b_n)$ and $b_n=\omega(a_n)$.
We say a sequence of events $\mathcal{A}_n$ holds with high probability if for any $c > 0$ there exists a finite constant $n_0$ depending only on $c$ such that $\mathbb{P}(\mathcal{A}_n)\geq 1-n^{-c}$ for all $n \geq n_0$.
We write $a_n = O_p(b_n)$ (resp. $a_n = o_p(b_n)$) to denote that $a_n = O(b_n)$ (resp. $a_n = o(b_n)$) holds with high probability.
We denote by $\mathcal{O}_d$ the set of $d \times d$ orthogonal
matrices. 
Given a matrix $\mathbf{M}$, we denote
its spectral, Frobenius, infinity norms, and the maximum entry (in modulus) by $\|\mathbf{M}\|$, $\|\mathbf{M}\|_{F}$, $\|\mathbf{M}\|_{\infty}$,  and $\|\mathbf{M}\|_{\max}$, respectively.

\section{Model}\label{sec:model}

In this section, we introduce our model for multiple dynamic networks with a clustering structure.
Suppose we have total $m$ dynamic networks $\{G_t^{(i)}\}_{i\in[m],t\in[T]}$ of $n$ common vertices for $T$ time points.
We build the model for multiple dynamic networks on the framework for the single dynamic network detailed in \cite{athreya2024discovering}.
Based on Random Dot Product Graph (RDPG), for each dynamic network $i\in[m]$, each vertex $s\in[n]$ at time $t\in[T]$ is associated with a latent position in $\mathbb{R}^d$, and edges between any two vertices arise independently with connection probability equal to the inner product of their respective latent position vectors. Therefore given latent position matrices $\{\mx_t^{(i)}\in\mathbb{R}^{n\times d}\}_{i\in[m],t\in[T]}$, we have $G_t^{(i)}\sim \text{RDPG}(\mx_t^{(i)})$ for all $i\in[m],t\in[T]$.

\begin{definition}[Random dot product graph \cite{athreya2018statistical}]
We say that the undirected random graph $G$ with adjacency matrix $\ma\in\mathbb{R}^{n\times n}$ is a random dot product graph with latent position matrix $\mx=[\mathbf{x}_1,\dots,\mathbf{x}_n]^\top\in\mathbb{R}^{n\times d}$ where each $\mathbf{x}_s$ denotes the latent position for vertex $s$ satisfying $\langle\mathbf{x}_s, \mathbf{x}_t\rangle=\mathbf{x}_s^\top \mathbf{x}_t\in[0,1]$ for all $s,t\in[n]$, and write $G\sim \text{RDPG}(\mx)$, if
$$\mathbb{P}[\ma|\mx]=\prod_{s<t}\langle\mathbf{x}_s, \mathbf{x}_t\rangle^{\ma_{s,t}}
\cdot(1-\langle\mathbf{x}_s, \mathbf{x}_t\rangle)^{1-{\ma_{s,t}}}.$$
We call $\mpp=\mx\mx^\top$ the connection probability matrix.
In this case, each $\ma_{s,t}$ is marginally distributed as Bernoulli$(\mpp_{s,t})$ where $\mpp_{s,t}=\langle\mathbf{x}_s, \mathbf{x}_t\rangle$.
\end{definition} 

RDPGs are a special case of latent
position graphs or graphons
\cite{Hoff2002,diaconis08:_graph_limit_exchan_random_graph,lovasz12:_large}.
In the general latent position graph model, each vertex $s$ is
associated with a latent or unobserved vector $\mathbf{x}_s\in\mathcal{X}$ where $\mathcal{X}$ is some latent space such as $\mathbb{R}^d$, and given the
collection of latent vectors $\{\mathbf{x}_s\}$, the edges are conditionally
independent Bernoulli random variables with probability $\mpp_{s,t}
= \kappa(\mathbf{x}_s, \mathbf{x}_t)$ for some kernel function $\kappa:\mathcal{X}\times \mathcal{X}\to[0,1]$.
For RDPGs, $\kappa$ is the inner product.

\begin{remark}[Orthogonal nonidentifiability in RDPGs]\label{rm:noni RDPG}
Note that if $\mx\in\mathbb{R}^{n\times d}$ is a latent position matrix, then for any $\mw\in \mathcal{O}_d$, $\mx$ and $\mx\mw$ give rise to RDPGs with the same probability matrix $\mpp=\mx\mx^\top=(\mx\mw)(\mx\mw)^\top$. Thus, the RDPG model has a nonidentifiability up to orthogonal transformation.
\end{remark}

We first consider latent position matrices $\{\mx^{(i)}_t\}$ as fixed to construct the model for simplicity and flexibility. We are also interested in the setting with randomly generated $\{\mx^{(i)}_t\}$ as described in \cite{athreya2024discovering}. For example, the latent position of each vertex can be generated from a random process \cite{chen2024euclidean}.
The model and theoretical results extended for random $\{\mx^{(i)}_t\}$ are discussed in Section~\ref{sec:random X}.

The key idea of \cite{athreya2024discovering} is to find a Euclidean analogue, called a \textit{mirror}, which is a finite-dimensional curve that retains important signal of the evolution patterns in dynamic networks. 
More specifically, a mirror of a dynamic network of RDPG model with latent position matrix $\{\mx_t\}_t$ is a dynamic vector $\{\mathbf{m}_t\}_t$ in a low-dimensional space, and the mirror keeps the overall trend of the dynamic network with that for any time points $t_1$ and $t_2$, the distance between $\mathbf{m}_{t_1}$ and $\mathbf{m}_{t_2}$ in the low-dimensional space is close to the ``distance" between $\mx_{t_1}$ and $\mx_{t_2}$.

With the above idea, for each dynamic network $i$ for $T$ time points, we formally construct the distance matrix $\md^{(i)}\in\mathbb{R}^{T\times T}$ to capture the change of latent position matrices across the $T$ time points, and have the corresponding mirror in $r$-dimensional space $\mm^{(i)}=[\mathbf{m}^{(i)}_1,\dots,\mathbf{m}^{(i)}_T]^\top \in\mathbb{R}^{T\times r}$ as following.

For each fixed $i\in[m]$, consider the dynamic network $\{G^{(i)}_t\}_{t\in[T]}$ of $n$ vertices for $T$ time points and $G^{(i)}_t\sim\text{RDPG}(\mx^{(i)}_t)$ for all $t\in[T]$.
Considering the orthogonal nonidentifiability of RDPGs as mentioned in Remark~\ref{rm:noni RDPG}, for any pair of time points $t_1$ and $t_2$ in $[T]$, we define a distance to measure the dissimilarity between the two latent position matrices $\mx^{(i)}_{t_1}$ and $\mx^{(i)}_{t_2}$ as
\begin{equation}\label{eq:Di_t1t2}
	\md^{(i)}_{t_1,t_2}:=\frac{1}{\sqrt{n}}\min_{\mo\in\mathcal{O}_d}\|\mx^{(i)}_{t_1}\mo-\mx^{(i)}_{t_2}\|_F,
\end{equation}
where $n^{-1/2}$ is used to compute the dissimilarity for the row/vertex average.
We then obtain the distance matrix $\md^{(i)}\in\mathbb{R}^{T\times T}$.



Then we find the configuration in a low-dimensional space to approximately preserve the dissimilarity across $\{\mx^{(i)}_t\}$ based on the distance matrix $\md^{(i)}$. That is, we seek $\mm^{(i)}=[\mathbf{m}^{(i)}_{1},\dots,\mathbf{m}^{(i)}_{T}]^\top\in\mathbb{R}^{T\times r}$ where the rows $\{\mathbf{m}^{(i)}_{t}\}_{t\in[T]}$ represent coordinates of points in $\mathbb{R}^{r}$ for some integer $r\geq 1$ such that their pairwise distances are ``as close as possible" to the distances given by $\md^{(i)}$, i.e. $$\|\mathbf{m}^{(i)}_{t_1}-\mathbf{m}^{(i)}_{t_2}\|\approx\md^{(i)}_{t_1,t_2}$$ for all $t_1,t_2\in[T]$. Such low-rank configuration can be derived by classical multidimensional scaling (CMDS) \cite{borg2007modern,wickelmaier2003introduction,li2020central} using the rank $r$ eigendecomposition of a matrix $\mb^{(i)}$, which is obtained from $\mathbf{D}^{(i)}$ by double centering.

\begin{definition}[CMDS embedding to dimension $r$]
Given a distance matrix $\md\in\mathbb{R}^{T\times T}$  and an embedding dimension $r$.
	We first apply double centering to compute $$\mb:=-\frac{1}{2}\mj\md^{\circ 2}\mj,$$ where $\md^{\circ 2}$ is $\md$ matrix entry-wise squared, and $\mj:=\mi-\frac{\mathbf{11}^\top}{T}$ is the centering matrix.
Then the CMDS embedding to dimension $r$ is $\mm=\muu\mLambda^{1/2},$ where $\muu\in\mathbb{R}^{T\times r}$ and the diagonal matrix $\mLambda^{(i)}\in\mathbb{R}^{r\times r}$ contains the $r$ leading eigenvectors and eigenvalues of $\mb$, respectively.
\end{definition}

\begin{remark}[Orthogonal nonidentifiability in CMDS embeddings]\label{rm:noni CMDS}
Note that the resulting configuration $\mm$ obtained by CMDS centers all points $\{\mathbf{m}_t\}_{t\in[T]}$ around the origin, resulting in an inherent issue of identifiability: $\mm$ is unique only up to an orthogonal transformation.
\end{remark}

For the multiple dynamic networks clustering problem of interest, we suppose the total $m$ dynamic networks belong to $K$ cluster, and denote their cluster labels by $\{Y_i\}_{i\in[m]}$ where each label $Y_i\in[K]$. 
For the dynamic networks in the same cluster, we suppose they share the same distance matrix, i.e. for any $i,j\in[m]$ such that $Y_i=Y_j$ we have $\md^{(i)}=\md^{(j)}$,  which means that they have similar change patterns of latent positions of the vertices across the time points. A special case is that the two dynamic networks share all the probability matrices across the time points and Theorem~\ref{thm:same D} shows that in this case the distance matrices for the latent positions are always identical.
Notice that our model is flexible to describe situations as long as the latent positions of the vertices exhibit similar patterns of variation across the time points, and the case in Theorem~\ref{thm:same D} is just a special case.


\begin{theorem}\label{thm:same D}
	If two dynamic networks $i,j\in[m]$ share the same probability matrices, i.e. $\mpp^{(i)}_t=\mpp^{(j)}_t$ for all $t\in[T]$, they then have the same distance matrix $\md^{(i)}=\md^{(j)}$. 
\end{theorem}

As we mentioned in Remark~\ref{rm:noni CMDS}, the mirror is unique only up to an orthogonal transformation. We have Theorem~\ref{thm:same M} to describe the relationship between the mirrors of dynamic networks in the same clusters with the same distance matrix.

\begin{theorem}\label{thm:same M}
	If two dynamic networks $i,j\in[m]$ have the same distance matrix $\md^{(i)}=\md^{(j)}$,  and we further suppose $\lambda_r(\mb^{(i)})>\lambda_{r+1}(\mb^{(i)})$, then their mirror matrices in $\mathbb{R}^r$ are the same up to an $r\times r$ orthogonal matrix, i.e. there exists $\mw_\mm^{(i,j)}\in\mathcal{O}_r$ such that $\mm^{(i)}\mw_\mm^{(i,j)}=\mm^{(j)}$.
\end{theorem}

Theorem~\ref{thm:same M} shows that in the same cluster, the dynamic networks have the same mirror up to orthogonal matrices. 
On the other hand, since CMDS, as a dimensionality reduction method, inevitably compresses some signals, it is theoretically possible to have the same mirror from different distance matrices. We are more interested in the situation 
where mirrors can differentiate between the evolution patterns corresponding to different clusters.
We suppose there exists a  distinct parameter $c>0$ such that for any two dynamic networks $i,j\in[m]$ with $Y_i\neq Y_j$, we have 
\begin{equation}\label{eq:define c}
	\frac{1}{\sqrt{T}}\min_{\mo\in\mathcal{O}_r}\|\mm^{(i)}\mo-\mm^{(j)}\|_F\geq c.
\end{equation}

\section{Algorithm}\label{sec:alg}

Our algorithm for clustering the evolution patterns of dynamic networks is based on the dissimilarities between mirrors.
To measure the dissimilarity of mirrors for these $m$ dynamic networks,  for any pair $i,j\in[m]$, we define the distance between two mirrors $\mm^{(i)}$ and $\mm^{(j)}$ as
$$\md^{\star}_{i,j}=\frac{1}{\sqrt{T}}\min_{\mo\in\mathcal{O}_r}\|\mm^{(i)}\mo-\mm^{(j)}\|_F,$$
and then obtain the pairwise mirror distance matrix $\md^\star\in\mathbb{R}^{m\times m}$.
By Theorem~\ref{thm:same M} and Eq.~\eqref{eq:define c}, $\md^\star$ exhibits a distinct block structure corresponding to the cluster assignments, i.e.,
\begin{equation}\label{eq:def Dstar}
	\md^{\star}_{i,j}
	\left\{
	\begin{aligned}
	&=0 \text{ if }Y_i=Y_j,\\
	&\geq c \text{ if }Y_i\neq Y_j
	\end{aligned}
	\right.
\end{equation}
for some $c>0$.

From the observed dynamic adjacency matrices $\{\ma^{(i)}_t\}_{i\in[m],t\in[T]}$, we can get an estimate of $\md^{\star}$, and this estimate $\hat\md^{\star}$ also has the similar block structure. Then the hierarchical clustering algorithm \cite{nielsen2016hierarchical,ward1963hierarchical} can be applied to the distance matrix $\hat\md^\star$ to restore the clusters of dynamic networks.
To obtain the estimated pairwise mirror distance matrix $\hat\md^{\star}$ from observed $\{\ma^{(i)}_t\}_{i\in[m],t\in[T]}$, we first compute the consistent estimates for the underlying unknown latent position matrices $\{\mx^{(i)}_t\}_{i\in[m],t\in[T]}$ by the rank-$d$ eigendecomposition of the adjacency matrices \cite{tang2018limit,sussman2012consistent}. More specifically, each $\hat\mx^{(i)}_t$ is the adjacency spectral embedding obtained from its respective adjacency matrix $\ma^{(i)}_t$.

\begin{definition}[Adjacency spectral embedding]
Given an observed adjacency matrix $\ma\in\mathbb{R}^{n\times n}$, we define the adjacency spectral embedding with dimension $d$ as $\hat\mx=\hat\muu|\hat\mLambda|^{1/2}$, where $\hat\mLambda\in\mathbb{R}^{d\times d}$ is the diagonal matrix containing the $d$ largest eigenvalues of $\ma$ and the columns of $\hat\muu\in\mathbb{R}^{n\times d}$ are the corresponding eigenvectors. $\hat\mx$ is the estimated latent position matrix.
\end{definition}

Based on the estimated latent position matrices $\{\hat\mx^{(i)}_t\}_{i\in[m],t\in[T]}$, with the similar procedure to have $\md^\star$ based on $\{\mx^{(i)}_t\}_{i\in[m],t\in[T]}$, the estimated pairwise mirror distance matrix $\hat\md^{\star}$ is obtained. See Algorithm~\ref{Alg_clustering} for details.

\begin{algorithm}[htbp]
\caption{Dynamic Network Clustering through Mirror Distance (DNCMD)}	
\label{Alg_clustering}
\begin{algorithmic}
\REQUIRE $m$ dynamic networks for $T$ time points $\{\ma^{(i)}_t\in\mathbb{R}^{n\times n}\}_{i\in[m],t\in[T]}$; an embedding dimension $d$ for networks; a dimension $r$ for mirror/low-rank configuration for the distance matrices; a cluster number $K$.

\textbf{Step 1: Constructing the mirror for each dynamic network $i\in[m]$}
\FOR{each dynamic network $i\in[m]$}
    \FOR{each time point $t\in[T]$}
        \STATE Compute $\hat\mx^{(i)}_t=\hat\muu^{(i)}_t|\hat\mLambda^{(i)}_t|^{1/2}\in\mathbb{R}^{n\times d}$ via eigendecomposition of $\ma^{(i)}_t$
    \ENDFOR
    \STATE Construct distance matrix $\hat\md^{(i)}\in\mathbb{R}^{T\times T}$ with entries $\hat\md^{(i)}_{t_1,t_2}=\frac{1}{\sqrt{n}}\min_{\mo\in\mathcal{O}_d}\|\hat\mx^{(i)}_{t_1}\mo-\hat\mx^{(i)}_{t_2}\|_F$
    \STATE Apply CMDS to compute mirror matrix $\hat\mm^{(i)}\in\mathbb{R}^{T\times r}$ from $\hat\md^{(i)}$
\ENDFOR

 \textbf{Step 2: Clustering dynamic networks}
\STATE Construct pairwise mirror distance matrix $\hat\md^\star\in\mathbb{R}^{m\times m}$ with $\hat\md^\star_{i,j}=\frac{1}{\sqrt{T}}\min_{\mo\in\mathcal{O}_r}\|\hat\mm^{(i)}\mo-\hat\mm^{(j)}\|_F$
\STATE Apply hierarchical clustering to $\hat\md^\star$ with $K$ clusters to obtain $\{\hat Y_i\}_{i\in[m]}$

\ENSURE The estimated cluster labels $\{\hat Y_i\}_{i\in[m]}$.
\end{algorithmic}
\end{algorithm}

In practice, we can choose the dimensions $d$ and $r$ by looking at the eigenvalues of $\{\ma^{(i)}_t\}$ and $\{\md^{(i)}\}$. A ubiquitous and principled method is to examine the so-called scree plot and look for ``elbow" and ``knees" defining the cut-off between the top (signal) $d$ or $r$ dimensions and the noise dimensions. 
\cite{zhu2006automatic} provides an automatic dimensionality selection procedure to look for the ``elbow'' by maximizing a profile likelihood function.
 \cite{han2023universal} suggests another universal approach to rank inference via residual subsampling for estimating the rank.
 We can also the determine by eigenvalue ratio test \cite{ahn2013eigenvalue} or by empirical distribution of eigenvalues \cite{onatski2010determining}. 
Note that the hierarchical clustering algorithm does not require a predefined number of clusters and can provide a dendrogram, a tree-like chart that illustrates the sequences of merges or splits of clusters. When two clusters are merged, the dendrogram connects them with a line, and the height of the connection represents the distance between those clusters. The optimal number of clusters $K$ can be chosen based on hierarchical structure of the dendrogram, and then be used to obtain clustering result $\{\hat Y_i\}$.

Note that Algorithm~\ref{Alg_clustering} is inherently a distributed algorithm and can benefit from the advantages of distributed algorithms \cite{fan2019distributed,hector2021distributed,chen2022distributed}.
The network $\{\ma^{(i)}_t\}$ can be stored on separate local machines.
Each local node can perform Step~1 in Algorithm~\ref{Alg_clustering} to obtain the $n\times d$ estimated latent position matrix $\hat\mx^{(i)}_t$ and send it to a central node. The central node collects $\{\hat\mx^{(i)}\}$ and then proceeds with the subsequent steps.
With this distributed procedure, the central node does not need the capability to store enormous amounts of data, which is especially important considering that real-world networks may contain a vast number of vertices.
And between nodes, only relatively small-scale matrices need to be transferred instead of the raw data, so that the communication cost has been significantly reduced, and the privacy of the raw information has been protected.
\section{Theoretical Results}
\label{sec:thm}

We now investigate the theoretical properties of our proposed clustering algorithm for $m$ dynamic networks $\{G^{(i)}\}_{i \in [m], t \in [T]}$, each with $n$ common vertices over $T$ time points and cluster structures as described in Section~\ref{sec:model}.
Our main result, presented in Theorem~\ref{thm:Dstar error}, establishes the error rate of the estimated pairwise mirror distance matrix $\hat\md^\star$, demonstrating that $\hat\md^\star$ closely resembles $\md^\star$ and thus also captures the block structure corresponding to the clustering assignments, thereby validating the effectiveness of DNCMD. Theorem~\ref{thm:Dstar error} is derived based on the estimation error results for the estimated distance matrix $\hat\md^{(i)}$ and the corresponding estimated mirror matrix $\hat\mm^{(i)}$ in Theorem~\ref{thm:D error} and Theorem~\ref{thm:M error} for each dynamic network $i \in [m]$.


\begin{theorem}\label{thm:D error}
For a fixed $i\in[m]$, consider a dynamic network $\{G^{(i)}_t\}_{t\in[T]}$ as defined in Section~\ref{sec:model}.
Suppose that 1) $n\rho_n \gg \log n$, where $\rho_n\in(0,1]$ is a sparsity factor  depending on $n$ such that for each $i\in[m]$ and each $t\in[T]$, $\|\mpp^{(i)}_t\|=\Theta(n\rho_n)$;
2) $\{\mpp^{(i)}_t\}$ have bounded condition number, i.e. there exists a finite constant $M>0$ such that
		$\max_{i\in[m],t\in[T]}\frac{\lambda_1(\mpp^{(i)}_t)}{\lambda_d(\mpp^{(i)}_t)}\leq M,$
		where $\lambda_1(\mpp^{(i)}_t)$ (resp. $\lambda_d(\mpp^{(i)}_t)$) denotes the largest (resp. smallest) eigenvalue of $\mpp^{(i)}_t$.
We then have
	$$\|\hat\md^{(i)}-\md^{(i)}\|_{\max}\lesssim  n^{-1/2}$$
	with high probability.
\end{theorem}

Condition~(1) in the statement of Theorem~\ref{thm:D error} is on the sparsity of the graph. We can interpret $n\rho_n$ as the growth rate for the average degrees of the graphs $\{\ma^{(i)}_t\}$ generated from $\{\mpp^{(i)}_t\}$, and Condition~(1) requires that the graphs are not too sparse. 
Condition (2) pertains to the homogeneity of \(\mpp^{(i)}\) and, consequently, also to the homogeneity of the latent positions. This condition can be relaxed; the bound $M$ does not necessarily have to be constant and can depend on $n$. In this case, the error bound for $\hat{\md}^{(i)}$ will involve $M$.
We further establish results for the estimated mirror matrix $\hat\mm^{(i)}$, based on Theorem~\ref{thm:D error}.

\begin{theorem}\label{thm:M error}
Consider the setting of Theorem~\ref{thm:D error} and further suppose that $\lambda_r(\mb^{(i)})=\omega(Tn^{-1}(n\rho_n)^{1/2})$ and $\lambda_{r+1}(\mb^{(i)})=O(Tn^{-1}(n\rho_n)^{1/2})$. We then have
\begin{align*}
     \frac{1}{\sqrt{T}}\min_{\mo\in\mathcal{O}_r}\|\hat\mm^{(i)}\mo-\mm^{(i)}\|_F\lesssim
 \frac{T^{1/2}(n\rho_n)^{1/2}\lambda_1^{1/2}(\mb^{(i)})}{n \lambda_r(\mb^{(i)})}\Big(1+\frac{T(n\rho_n)^{1/2}\lambda_1^{1/2}(\mb^{(i)})}{n \lambda_r^{3/2}(\mb^{(i)})}\Big)
\end{align*}
 with high probability.
 If we further assume $\mb^{(i)}$ has bounded condition number, i.e. there exists a finite constant $M'>0$ such that $\frac{\lambda_1(\mb^{(i)})}{\lambda_r(\mb^{(i)})}\leq M'$, we then have
\begin{equation}\label{eq:MO-M}
	 \frac{1}{\sqrt{T}}\min_{\mo\in\mathcal{O}_r}\|\hat\mm^{(i)}\mo-\mm^{(i)}\|_F\lesssim
 \frac{T^{1/2}(n\rho_n)^{1/2}}{n \lambda_r^{1/2}(\mb^{(i)})}
\end{equation}
 with high probability.
\end{theorem}

The condition about the eigenvalues of the doubly centered distance matrix $\mb^{(i)}$ in Theorem~\ref{thm:M error} implies that its largest $r$ eigenvalues and their corresponding eigenvectors contain the primary signal. This further indicates that the $r$-dimensional configuration of these latent position matrices preserves the primary signal.

We then analyze the error of the estimated pairwise mirror distance matrix $\hat\md^\star$ in Theorem~\ref{thm:Dstar error} and the resulting clustering performance based on its structure.
Theorem~\ref{thm:Dstar error} shows that the clustering result of DNCMD achieves exact recovery and provides a lower bound for the distinct parameter $c$. 
Here we use the terminology exact recovery when the partition is recovered correctly with high probability.

\begin{theorem}\label{thm:Dstar error}
	Consider the setting of Theorem~\ref{thm:M error} for all $i\in[m]$. We suppose that there exists a parameter $\lambda$ such that $\lambda_r(\mb^{(i)})\asymp \lambda$ for all $i\in[m]$.
	We then have
	$$
	\|\hat\md^{\star}-\md^{\star}\|_{\max}
	\lesssim \frac{T^{1/2}(n\rho_n)^{1/2}}{n \lambda^{1/2}}
	$$
	with high probability.
	We further suppose that the  distinct parameter $c>0$ defined in Eq.~\eqref{eq:define c}, which distinguishes the mirrors for different clusters, satisfies $c=\omega(\frac{T^{1/2}(n\rho_n)^{1/2}}{n \lambda^{1/2}})$. We then have
	\begin{equation}\label{eq:clear distance}
		\hat\md^{\star}_{i,j}
	\left\{
	\begin{aligned}
	&< c/2 \text{ if }Y_i=Y_j,\\
	&> c/2 \text{ if }Y_i\neq Y_j
	\end{aligned}
	\right.
\end{equation}
with high probability, and hierarchical clustering for $\hat\md^\star$ achieves exact recovery. 
\end{theorem}

The error bounds in Theorem~\ref{thm:M error} and Theorem~\ref{thm:Dstar error} depend on the eigenvalues of $\{\mb^{(i)}\}$, and the scales of $\{\mb^{(i)}\}$ may increase with $T$ as $\{\mb^{(i)}\}$ are of size $T \times T$.
Actually the relationship $\lambda_r(\mb^{(i)})\asymp T\rho_n\tilde\rho$ can be derived under a general case as described in following Proposition~\ref{prop:M error2}. Recall that $\rho_n$ is the parameter indicating the sparsity of the networks or the magnitude of the latent positions, and here we further use a parameter $\tilde{\rho}$ to describe the magnitude of $\md^{(i)}$ or the change of latent positions across time points. 

Proposition~\ref{prop:M error2} and the subsequent Proposition~\ref{prop:Dstar error2} discuss a general case, providing an intuitive understanding of the conditions under which $\hat\mm^{(i)}$ converges to the true $\mm^{(i)}$, $\hat\md^\star$ converges to the true $\md^\star$, and the requirements for DNCMD to achieve exact recovery.

\begin{proposition}\label{prop:M error2}
	Consider the setting of Theorem~\ref{thm:M error}. Notice $\|\mx_t^{(i)}\|_F^2\asymp n\rho_n$, and we further suppose there exists a factor $\tilde \rho>0$ to describe the magnitude of the change of $\{\mx_t^{(i)}\}_{t\in[T]}$ such that $\min_{\mo\in\mathcal{O}_d}\|\mx^{(i)}_{t_1}\mo-\mx^{(i)}_{t_2}\|_F^2\asymp n\rho_n \tilde \rho$ for $t_1\neq t_2$. We then have $\lambda_r(\mb^{(i)})\asymp T\rho_n\tilde\rho$ and thus
	$$
	\frac{1}{\sqrt{T}}\min_{\mo\in\mathcal{O}_r}\|\hat\mm^{(i)}\mo-\mm^{(i)}\|_F
    \lesssim \frac{1}{(n\tilde\rho)^{1/2}}
	$$
	with high probability. We also have $\frac{1}{\sqrt{T}}\|\mm^{(i)}\|_F\asymp (\rho_n\tilde \rho)^{1/2}$. Then, under the assumptions $n\rho_n\tilde\rho^2=\omega(1)$ and $T=O(n^{\ell})$ for some $\ell>0$, there exists $\mw^{(i)}_\mm\in\mathcal{O}_r$ such that $$\hat\mm^{(i)}\mw^{(i)}_\mm\xrightarrow{w.h.p.} \mm^{(i)}$$
	as $n\rightarrow\infty$, where $\xrightarrow{w.h.p.}$ denotes  convergence with high probability.
\end{proposition}

In Proposition~\ref{prop:M error2} we suppose $\min_{\mo\in\mathcal{O}_d}\|\mx^{(i)}_{t_1}\mo-\mx^{(i)}_{t_2}\|_F^2\asymp n\rho_n \tilde \rho$ for all pairs $t_1\neq t_2$ just for ease of exposition. The proof can be adapted to allow $\min_{\mo\in\mathcal{O}_d}\|\mx^{(i)}_{t_1}\mo-\mx^{(i)}_{t_2}\|_F^2=0$ for a certain percentage of pairs, thereby handling the more general case.

Proposition~\ref{prop:M error2} shows that the estimated mirror $\hat\mm^{(i)}$ converges to the true mirror $\mm^{(i)}$ with high probability when the average degree of graphs $n\rho_n$ or the change magnitude $\tilde\rho$ is large enough, i.e., $(n\rho_n)\tilde\rho^2=\omega(1)$, and the number of time points does not grow too fast, i.e. there exists $\ell>0$ such that $T=O(n^{\ell})$. The discussion in Proposition~\ref{prop:M error2} intuitively demonstrates under which conditions our estimator of the mirror matrix is asymptotically convergent.



\begin{proposition}\label{prop:Dstar error2}
Consider the setting of Proposition~\ref{prop:M error2} for all $i\in[m]$. We then have
$$
	\|\hat\md^{\star}-\md^{\star}\|_{\max}
	\lesssim \frac{1}{(n\tilde\rho)^{1/2}}
	$$
	with high probability. 
	We further suppose the distinct parameter $c$ satisfies $c=\omega((n\tilde\rho)^{-1/2})$. 
	Then Eq.~\eqref{eq:clear distance} for $\hat\md^\star$ holds, and hierarchical clustering for $\hat\md^\star$ achieves exact recovery.
\end{proposition}

	
	Proposition~\ref{prop:Dstar error2} shows the exact recovery of DNCMD requires $c=\omega((n\tilde\rho)^{-1/2})$.
	Recall that in this setting, for individual mirror matrices, we obtain $\frac{1}{\sqrt{T}}\|\mm^{(i)}\|_F\asymp (\rho_n\tilde \rho)^{1/2}$, and for pairs of mirrors from two different clusters with $Y_i\neq Y_j$, we have $\frac{1}{\sqrt{T}}\min_{\mo\in\mathcal{O}_r}\|\mm^{(i)}\mo-\mm^{(j)}\|_F\geq c$. 
	It follows that, the block structure of $\hat\md^\star$ in Eq.~\eqref{eq:clear distance} and the resulting exact recovery of DNCMD require that the ratio of the dissimilarity between clusters to the strength of mirrors satisfies
	$$
	    \frac{\frac{1}{\sqrt{T}}\min_{\mo\in\mathcal{O}_r}\|\mm^{(i)}\mo-\mm^{(j)}\|_F}{\frac{1}{\sqrt{T}}\|\mm^{(i)}\|_F}
	\gtrsim \frac{c}{(\rho_n\tilde \rho)^{1/2}}
	=\omega\Big(\frac{(n\tilde\rho)^{-1/2}}{(\rho_n\tilde \rho)^{1/2}}\Big)
	=\omega((n\rho_n)^{-1/2}\tilde\rho^{-1}).
	$$
	It means a larger average degree of graphs $n\rho_n$ or a larger change magnitude $\tilde\rho$ allows smaller relative dissimilarity between clusters.
	
	\begin{remark}
		In the discussion about Proposition~\ref{prop:Dstar error2}, we do not see that a larger number of time points is beneficial to the error rate. The same phenomenon also appears in Proposition~\ref{prop:M error2} for the estimation error for mirrors. It does not mean that we do not need to collect the dynamic networks for more time points. Too inadequate data collection can result in a loss of capture of signals. We consider the example in (a) of Figure~\ref{fig:for T}. If we just collect the data for time points $1,2,\dots,5$, the difference between dynamic network $1$ and $2$ cannot be captured. Adequate data collection is necessary for detecting discrepancies between dynamic networks, i.e. a large enough $T$ is necessary for accurate $\frac{1}{\sqrt{T}}\min_{\mo\in\mathcal{O}_r}\|\mm^{(i)}\mo-\mm^{(j)}\|_F$.
	On the other side, if $T$ is large enough to sufficiently detect the discrepancies between dynamic networks, by the theoretical analysis, we do not need to collect too much data. The example in (b) of Figure~\ref{fig:for T} is periodic and the data for one cycle can contain enough signals. 
	
\begin{figure}[htbp] 
\centering
\subfigure[]{%
\includegraphics[width=7.5cm]{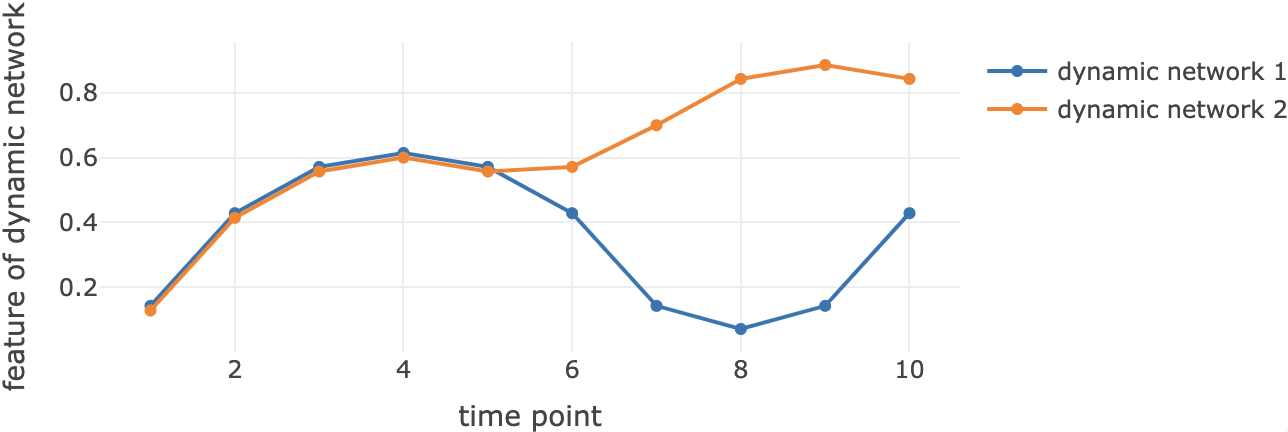}%
\label{fig:for_T_a}}
\subfigure[]{%
\includegraphics[width=7.5cm]{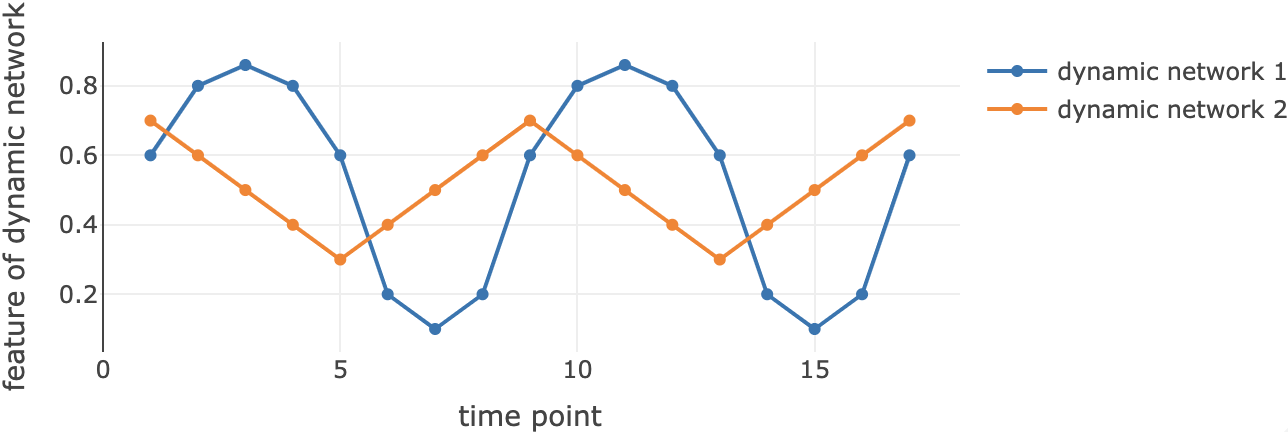}%
\label{fig:for_T_b}}
\caption{Examples for dynamic networks.}
\label{fig:for T}
\end{figure}
	\end{remark}

\subsection{Theoretical results for random latent positions}\label{sec:random X}

The model in Section~\ref{sec:model} and the analysis in Section~\ref{sec:thm} are for deterministic latent positions $\{\mx^{(i)}_t\}$.
Actually, the theoretical guarantees can also be provided for a case where $\{\mx^{(i)}_t\}$ are randomly generated from latent position stochastic processes as described in \cite{athreya2024discovering}.

In particular, for a dynamic network $\{G_t\}_{t}$ generated from RDPG models with a dynamic latent position matrix $\{\mx_t\}_{t}$, for each vertex $s\in[n]$, the $s$th row of $\{\mx_t\}_{t}$ is the corresponding dynamic latent position of this vertex. Suppose the latent positions for all vertices are i.i.d. samples of a multivariate stochastic process $X$ defined on a probability space $(\Omega,\mathcal{F},P)$; we refer to this as a latent position stochastic process.
For a particular sample point $\omega\in \Omega$, $X(t,\omega)\in\mathbb{R}^d$ is the realization of the associated latent position for this vertex at time $t$.
On the one hand, for a fixed $\omega$, $\{X(t,\omega)\}_{t}$ is the realized trajectory of a $d$-dimensional stochastic process.
On the other hand, for a given time $t$, $\{X(t,\omega)\}_{\omega\in\Omega}$ represents all possible latent positions at this time.
In order for the inner product to be a well-defined link function, we require that the distribution of $X(t,\cdot)$ follows an inner-product distribution for all $t$.
For notational simplicity, we will use $X(t,\cdot)$ and $X_t$ interchangeably.
\begin{definition}[Inner product distribution]
	Let $F$ be a probability distribution on $\mathbb{R}^d$.
	We say that $F$ is a inner product distribution, if $\mathbf{x}^\top \mathbf{x}'\in[0,1]$ for all $\mathbf{x},\mathbf{x}'\in \supp (F)$.
\end{definition}

For multiple dynamic networks $\{G^{(i)}_t\}_{i\in[m],t\in[T]}$, we denote the corresponding latent position stochastic processes by $\{X^{(i)}\}$.
For each fixed $i\in[m]$, as in the case of deterministic latent positions, we define the distance between any two latent position matrices $\mx^{(i)}_{t_1}$ and $\mx^{(i)}_{t_2}$ in Eq.~\eqref{eq:Di_t1t2} for different time points, here we also define a distance to quantify the difference between two latent position distributions $X_{t_1}$ and $X_{t_2}$.
For any pair of time points $t_1$ and $t_2$ in $[T]$, we define a distance between the two latent position distributions $X^{(i)}_{t_1}$ and $X^{(i)}_{t_2}$ as
\begin{equation}\label{eq:breve Dt1t2}
	\breve\md^{(i)}_{t_1,t_2}:=\min_{\mo\in\mathcal{O}_d} \|\mathbb{E}[(\mo X^{(i)}_{t_1}-X^{(i)}_{t_2})(\mo X^{(i)}_{t_1}-X^{(i)}_{t_2})^\top]\|^{1/2},
\end{equation}
and therefore we obtain the distance matrix $\breve\md^{(i)}\in\mathbb{R}^{T\times T}$.
In this case, to measure the sparsity of networks or the magnitude of latent position matrices, we define $\breve\mpp^{(i)}_t:=\mathbb{E}[X^{(i)}_tX^{(i)\top}_t]$, and suppose that $\breve\mpp^{(i)}_t$ has rank $d$ and $\|\breve\mpp^{(i)}_t\|\asymp \breve\rho_n$, where $\breve\rho_n$ is a sparsity factor.
Based on the distance matrix $\breve\md^{(i)}$, we have the associated mirror matrix $\breve\mm^{(i)}$, which is the CMDS embedding to dimension $r$ obtained from the distance matrix $\breve\md^{(i)}$ with the corresponding doubly centered matrix $\breve\mb^{(i)}:=-\frac{1}{2}\mj((\breve\md^{(i)})^{\circ 2})\mj^\top$.

For the multiple dynamic networks clustering problem of interest, the $m$ dynamic networks are divided into $K$ clusters with cluster labels $\{Y_i\}_{i\in[m]}$.
We suppose the dynamic networks in the same cluster generated from the same latent position stochastic process.
That is, for any $i,j\in[m]$ such that $Y_i=Y_j$, $X^{(i)}$ and $X^{(j)}$ follow the same distribution.
Then with Eq.~\eqref{eq:breve Dt1t2} we have $\breve\md^{(i)}=\breve\md^{(j)}$,
and further according to Theorem~\ref{thm:same M}, if $\lambda_r(\breve\mb^{(i)})>\lambda_{r+1}(\breve\mb^{(i)})$, the mirrors $\breve\mm^{(i)}$ and $\breve\mm^{(j)}$ are the same up to an $r\times r$ orthogonal matrix.
Recall that we are more interested in the situation where mirrors of different clusters can be distinguished. We suppose there exists a distinct parameter $\breve c>0$ such that for any two dynamic networks $i,j\in[m]$ with $Y_i\neq Y_j$, we have
$
	\frac{1}{\sqrt{T}}\min_{\mo\in\mathcal{O}_r}\|\breve\mm^{(i)}\mo-\breve\mm^{(j)}\|_F\geq \breve c.
$
We define the distance between mirrors as $\breve\md^\star_{i,j}=\frac{1}{\sqrt{T}}\min_{\mo\in\mathcal{O}_r}\|\breve\mm^{(i)}\mo-\breve\mm^{(j)}\|_F$, and then have the pairwise mirror distance matrix $\breve\md^\star$ with the clear block structure as shown in Eq.~\eqref{eq:def Dstar}.

We now analyze the theoretical results of DNCMD under this model with randomly generated latent positions.
The error bounds for $\hat\mm^{(i)}$ and $\hat\md^\star$ as the estimates of $\breve\mm^{(i)}$ and $\breve\md^\star$ is given in Theorem~\ref{thm:M error2} and Theorem~\ref{thm:Dstar error2}, respectively.
We want to emphasize that, under the assumption that for each dynamic network, the latent positions of all vertices are i.i.d. samples of a stochastic process, it is natural to allow different dynamic networks to have different numbers of vertices. Specifically, we can suppose dynamic network $i$ has $n_i$ vertices for all $i\in[m]$, and Algorithm~\ref{Alg_clustering} can be extended to such a scenario. For the following theoretical results, we let $n=\min_{i\in[m]}\{n_i\}$.
\begin{theorem}\label{thm:M error2}
Consider a dynamic network $\{G^{(i)}_t\}_{t\in[T]}$ for a fixed $i\in[m]$ as defined in Section~\ref{sec:random X}.
We suppose that $\lambda_r(\breve\mb^{(i)})=\omega(Tn^{-1}(n\breve\rho_n)^{1/2}\log n)$ and $\lambda_{r+1}(\breve\mb^{(i)})=O(Tn^{-1}(n\breve\rho_n)^{1/2}\log n).$ We then have
\begin{align*}
     \frac{1}{\sqrt{T}}\min_{\mo\in\mathcal{O}_r}\|\hat\mm^{(i)}\mo-\mm^{(i)}\|_F\lesssim
 \frac{T^{1/2}(n\breve\rho_n)^{1/2}(\log n)\lambda_1^{1/2}(\breve\mb^{(i)})}{n \lambda_r(\breve\mb^{(i)})}
	 \Big(1+\frac{T(n\breve\rho_n)^{1/2}(\log n)\lambda_1^{1/2}(\breve\mb^{(i)})}{n \lambda_r^{3/2}(\breve\mb^{(i)})}\Big)
\end{align*}
 with high probability.
 If we further assume $\breve\mb^{(i)}$ has bounded condition number, i.e. there exists a finite constant $\breve M'>0$ such that $\frac{\lambda_1(\breve\mb^{(i)})}{\lambda_r(\breve\mb^{(i)})}\leq \breve M'$, we then have
\begin{equation*}
	 \frac{1}{\sqrt{T}}\min_{\mo\in\mathcal{O}_r}\|\hat\mm^{(i)}\mo-\breve\mm^{(i)}\|_F\lesssim
 \frac{T^{1/2}(n\breve\rho_n)^{1/2}\log n}{n \lambda_r^{1/2}(\breve\mb^{(i)})}
\end{equation*}
 with high probability.
\end{theorem}

\begin{theorem}\label{thm:Dstar error2}
	Consider the setting of Theorem~\ref{thm:M error2}. We suppose there exists a parameter $\breve \lambda$ such that $\lambda_r(\breve\mb^{(i)})\asymp \breve\lambda$ for all $i\in[m]$.
	We then have
	$
	\|\hat\md^{\star}-\breve\md^{\star}\|_{\max}
	\lesssim \frac{T^{1/2}(n\breve\rho_n)^{1/2}\log n}{n \breve \lambda^{1/2}}
	$
	with high probability.
	We further suppose the distinct parameter $\breve c$ satisfies $$\breve c=\omega\left(\frac{T^{1/2}(n\breve\rho_n)^{1/2}\log n}{n \breve\lambda^{1/2}}\right).$$ 
	Then Eq.~\eqref{eq:clear distance} for $\hat\md^\star$ holds, and hierarchical clustering for $\hat\md^\star$ achieves exact recovery.
\end{theorem}
Recall that we suppose $\|\mathbb{E}[X^{(i)}_tX^{(i)\top}_t]\|\asymp \breve\rho_n$ with some sparsity parameter $\breve\rho$.
If we further suppose there exists a factor $\breve{\tilde\rho}>0$ to describe the magnitude of the change of $\{X^{(i)}_t\}_{t\in[T]}$ such that $\min_{\mo\in\mathcal{O}_d} \|\mathbb{E}[(\mo X^{(i)}_{t_1}-X^{(i)}_{t_2})(\mo X^{(i)}_{t_1}-X^{(i)}_{t_2})^\top]\|\asymp \breve\rho_n\breve{\tilde\rho}$ for $t_1\neq t_2$, the similar analysis as Proposition~\ref{prop:M error2} and Proposition~\ref{prop:Dstar error2} can be derived.

\section{Simulation Results}\label{sec:simu}

In this section, we perform numerical experiments to demonstrate the theoretical results from Section~\ref{sec:thm} in Section~\ref{sec:mirror estimate}, and show the clustering performance of DNCMD in Section~\ref{sec:performance}.

\subsection{Error in mirror estimate}
\label{sec:mirror estimate}

We demonstrate the theoretical results for the estimated mirror in Theorem~\ref{thm:M error} with numerical simulation.
Theorem~\ref{thm:M error} states a result for any fixed dynamic network $i\in[m]$.
For ease of exposition, in this section we will fix a value of $i\in[m]$ and thereby drop the index $i$ from our matrices.
We generate the true latent positions $\{\mx_t\}$ with the following random walk process. Recall that the true latent position of vertex $s\in[n]$ in time point $t\in[T]$ is the $s$th row of $\mx_t$. 
Let $d=1$ and denote the track of the latent position of vertex $s\in[n]$ as $\{X^{<s>}_t\in \mathbb{R}\}_{t\in[T]}$. 
Let $\tilde c\geq 0$ and $\delta_T>0$ be two constants satisfying $\tilde c+\delta_T T\leq 1$. 
For a fixed $p\in(0,1)$ and for all $s\in[n]$, we generate $\{X^{<s>}_t\in \mathbb{R}\}_{t\in[T]}$ independently as follows: $X^{<s>}_0=\tilde c$, and for $t\in[T]$,
\begin{equation}\label{eq:model p}
X^{<s>}_t= \begin{cases}
    X^{<s>}_{t-1}+\delta_T & \text{with probability } p, \\ 
    X^{<s>}_{t-1} & \text{with probability } 1-p.
\end{cases}
\end{equation}
When we set $\delta_T=(1-\tilde c)/T$, 
an observation is that, the dynamic network with the latent positions generated from this model has $\mb$ with $\lambda_1(\mb)\asymp T$ and $\lambda_k(\mb)\lesssim 1$ for $k=2,\dots, T$; see \cite{chen2024euclidean} for more details, and then we know $r=1$.
We demonstrate the error rate of $\hat\mm$ in Theorem~\ref{thm:M error} under this setting.
More specifically, we fix $\tilde c=0.1, p=0.4$, and either fix $n=100$ and vary $T\in\{20,30,40,50,70,100\}$ or fix $T=10$ and vary $n\in\{50,100,200,400,800\}$.
The estimation error of $\hat\mm$ is evaluated using $\frac{1}{\sqrt{T}}\min_{\mo\in\mathcal{O}_r}\|\hat\mm\mo-\mm\|_F$, and the results are summarized in Figure~\ref{fig:simulation_varyT}.
For this setting, 
 according to Theorem~\ref{thm:M error} and Proposition~\ref{prop:M error2}, if we only vary $T$, we have $\frac{1}{\sqrt{T}}\min_{\mo\in\mathcal{O}_r}\|\hat\mm\mo-\mm\|_F\lesssim 1$, and if we only vary $n$, we have $\frac{1}{\sqrt{T}}\min_{\mo\in\mathcal{O}_r}\|\hat\mm\mo-\mm\|_F\lesssim n^{-1/2}$. We see the changes of $\frac{1}{\sqrt{T}}\min_{\mo\in\mathcal{O}_r}\|\hat\mm\mo-\mm\|_F$ in Figure~\ref{fig:simulation_varyT} comply with the theoretical error rates obtained from Theorem~\ref{thm:M error} and Proposition~\ref{prop:M error2}.
\begin{figure}[htbp] 
\centering
\subfigure{%
\includegraphics[width=4.5cm]{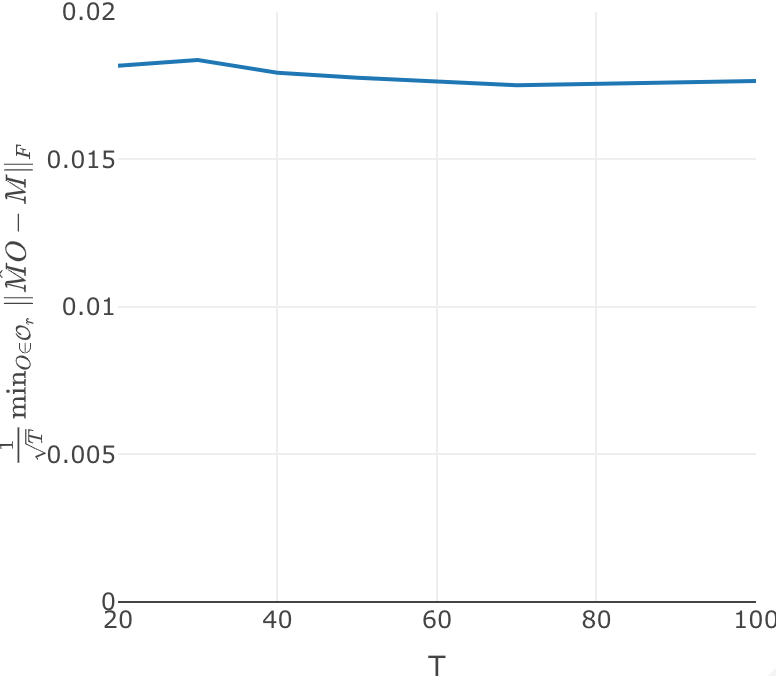}%
\label{fig:simulation_varyT_a}}
\subfigure{%
\includegraphics[width=4.5cm]{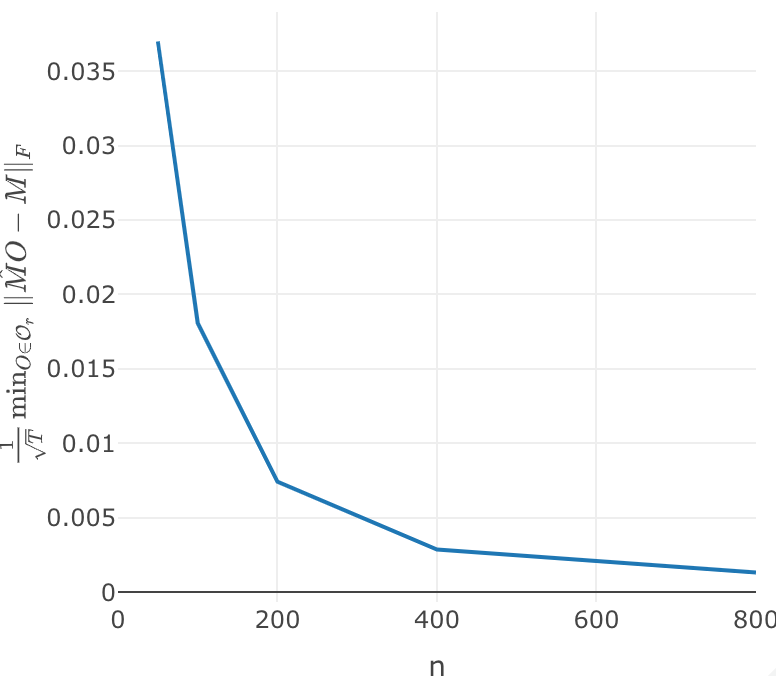}%
\label{fig:simulation_varyT_b}}
\caption{Empirical estimate rates for $ \frac{1}{\sqrt{T}}\min_{\mo\in\mathcal{O}_r}\|\hat\mm\mo-\mm\|_F$  as $T$ or $n$ changes.
	    Left panel: vary $T\in\{20,30,40,50,70,100\}$ while fixing $n=100$.
	    Right panel: vary $n\in\{50,100,200,400,800\}$ while fixing $T=10$.
	    The results are averaged over $100$ independent Monte Carlo replicates.}
\label{fig:simulation_varyT}
\end{figure}

\subsection{Clustering performance}

\label{sec:performance}

We perform a numerical experiment to show the clustering performance of DNCMD.
We consider $K=2$ clusters of the dynamic networks generated with Eq.~\eqref{eq:model p}.
These two clusters each have $20$ dynamic networks, and thus we need to cluster $m=2\times 20=40$ dynamic networks.
For the dynamic networks, we fix the number of time points $T=50$, and the middle time point $t=25$ is a change point of the probability of random walk $p$. 
More specifically, for cluster $1$, we set $p=0.45$ for $t=1,\dots,25$ and set $p=0.55$ for $t=26,\dots,50$.
And for cluster $2$, we set $p=0.55$ for $t=1,\dots,25$ and set $p=0.45$ for $t=26,\dots,50$.
According to Theorem~\ref{thm:Dstar error}, Proposition~\ref{prop:Dstar error2}, and Theorem~\ref{thm:Dstar error2}, we know DNCMD achieves exact recovery, and it means that for $n$ large enough DNCMD can recover the true clusters with high probability. We vary $n\in\{20,30,40,80,120,200\}$ to see the clustering performance of DNCMD.

To measure the clustering accuracy, we compute the adjusted rank index (ARI) for the similarity between the estimated cluster labels and the true labels.
The higher the ARI value, the closer the empirical clustering result and the true label assignment are to each other. ARI ranges from $-1$ to $1$, where $1$ indicates perfect agreement between the empirical clustering result and the true label assignment, and $0$ indicates a random agreement.
Since there is no existing clustering algorithm specifically designed for such dynamic network clustering problems, we compare the results of DNCMD with those of the classical clustering algorithm, k-means. More specifically, we consider using k-means to directly cluster the original adjacency matrices $\{\ma^{(i)}\}$, or using k-means to cluster the estimated low-rank latent position matrices $\{\hat\mx^{(i)}\}$.
We run $100$ independent Monte Carlo replicates and  summarize the result in Figure~\ref{fig:simulation2}. Figure~\ref{fig:simulation2} shows that, when $n$ is large enough (in this example, $n$ only needs to exceed about $100$), DNCMD always has the ARI very close to $1$ and can recover the true clusters perfectly. And the k-means algorithm applied to the original adjacency matrices or the latent position matrices is not very effective and stable for this problem.

\begin{figure}[htbp] 
\centering
\subfigure{%
\includegraphics[width=8cm]{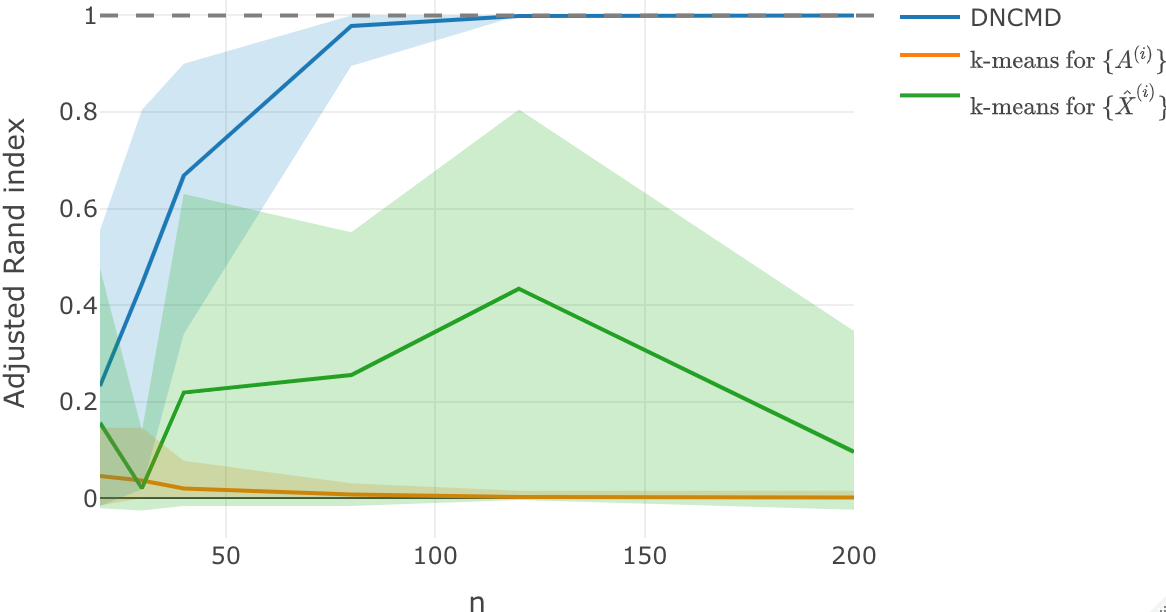}%
\label{fig:simulation2_a}}
\caption{Sample means and $90\%$ empirical confidence intervals of adjusted Rand index for DNCMD and the k-means algorithm based on $100$ independent Monte Carlo replicates. We tune $n\in\{20,30,40,80,120,200\}$. Other detailed settings can be found in Section~\ref{sec:performance}.}
\label{fig:simulation2}
\end{figure}


\section{Real Data Experiments}\label{sec:real}

In this section, we demonstrate the use of DNCMD on Drosophila larval connectome data to analyze edge functions in Section~\ref{sec:brain}, and cluster the trade dynamic networks to analyze the similarities and differences in trade evolution across different food products in Section~\ref{sec:trade}.

\subsection{Drosophila larval connectome data}\label{sec:brain}

The complete synaptic-resolution connectome of the Drosophila larval brain was recently completed \cite{winding2023connectome}, enabling the generation of biologically realistic models of its neural circuit based on known anatomical connectivity \cite{eschbach2020recurrent}.
Training such realistic models in simulations provides significant computational power and flexibility, enabling us to study and enhance our understanding of how real animal brains operate. This approach has been applied in recent papers \cite{acharyya2024consistent} to investigate the operational mechanisms of the Drosophila larval brain.

The connectivity-constrained model proposed in \cite{eschbach2020recurrent,jiang2021models} can be used to generate a recurrent network of the larval mushroom body with feedback neurons, where the connectome can be constrained, and in this process, a series of stimuli can be delivered to the network; see \cite{eschbach2020recurrent,jiang2021models} for more details. 
Therefore, by analyzing the dynamic networks generated by this model after the removal of particular edges between neurons, we can indirectly study the function of these edges.

In this experiment, we consider the process of learning a new association between a conditioned stimulus (CS) (for example, an odor) and rewards or punishments (unconditioned stimuli, US), which serve as reinforcement. We simulate the activity of the dynamic neuron network for a total of $T=160$ time points, and at $t=16$, a random odor is delivered to the neurons in the mushroom body (CS1), followed by a reward delivered at $t=20$. We set all stimuli to last for $3$ time points. After this initial CS of odor with the reward, the odor is also delivered again without the reward at $t=80$ (CS2) and $t=140$ (CS3). The association may be weakened by exposure to the same CS without reinforcement during CS2, so we may observe a decrease in network attraction from CS2 to CS3. Here, we define network attraction as the ratio of the strength of activity in neurons responsible for attraction to the strength of activity in neurons responsible for aversion. Thus, this process includes the learning of the association and its extinction. A schematic of this process can be seen in Figure~\ref{fig:process}.
\begin{figure}[htbp]
\centering
\includegraphics[width=7cm]{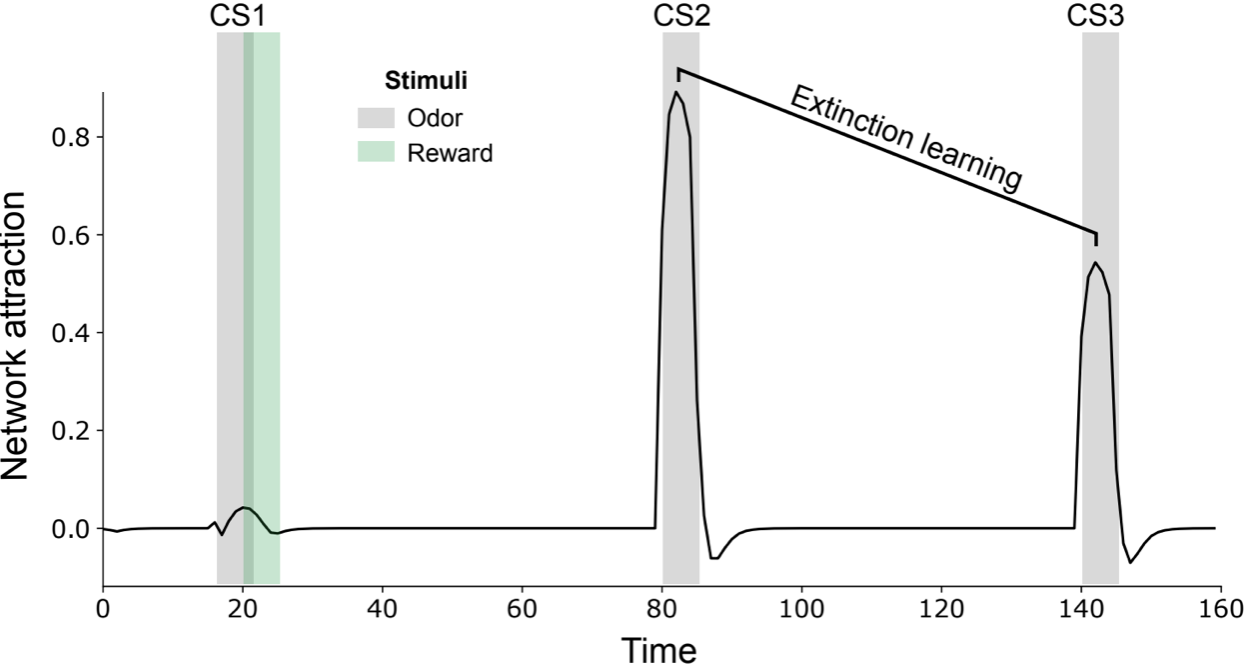}
\caption{Schematic of the process including the learning of an association and its extinction.}
\label{fig:process}
\end{figure}

We study the edges that may have a significant and specific impact on the aforementioned process.
We consider $13$ edges that may have an impact. To study their specific impacts, for each one removed, we generate the corresponding dynamic network of neurons for the aforementioned process using the connectivity-constrained model.
For removal of each single edge, we obtain $11$ replications using different randomization seeds, resulting in a total of $m=13 \times 11=143$ dynamic networks, where each network has $n=140$ nodes.
We can expect that if most replicates resulting from the removal of a particular edge are highly similar to each other and distinctly different from the dynamic networks generated by removing other edges, forming a distinct cluster, it would indicate that this edge plays a critical role in the process of building associations and has a highly specific impact.

We applied the DNCMD algorithm to obtain the dendrogram shown in Figure~\ref{fig:dendrogram}, with the embedding dimension for vertex latent positions set to $d=5$ and the dimension for mirror set to $r=3$. 
In Figure~\ref{fig:dendrogram}, we can clearly see that most of the dynamic networks resulting from the removal of edge $6$ form a distinct cluster, indicating that edge $6$ plays a critical and unique role in this process.
\begin{figure*}[htbp] 
\centering
\subfigure{%
\includegraphics[width=0.92\textwidth]{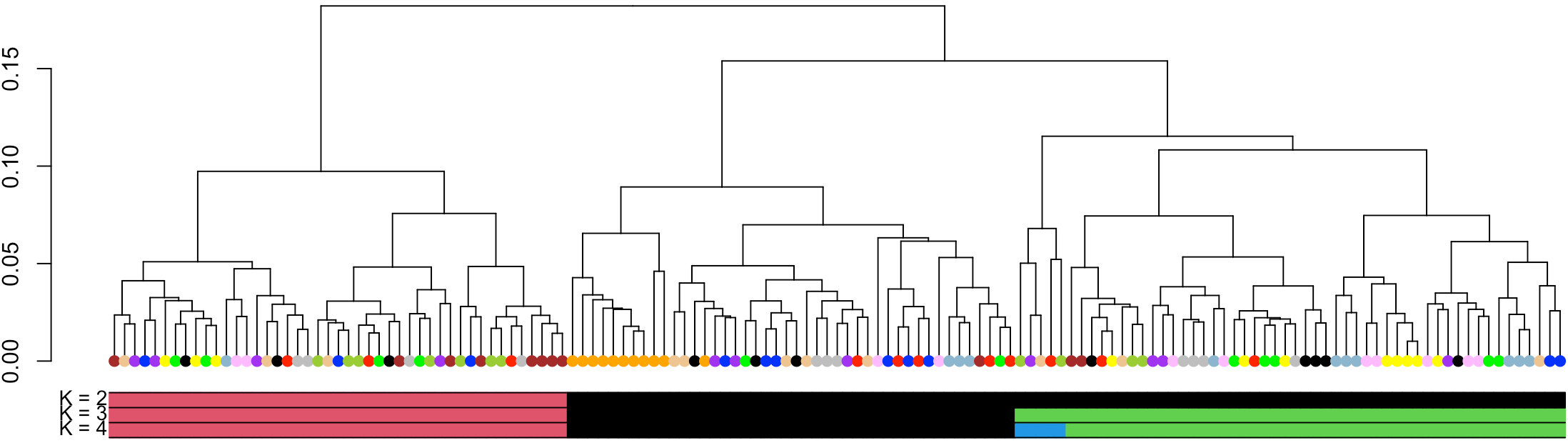}%
\label{fig:dendrogram_a}}
\hfil
\subfigure{%
\includegraphics[width=0.03\textwidth]{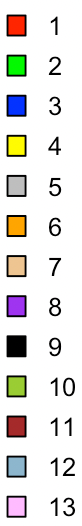}%
\label{fig:dendrogram_b}}
\caption{Dendrogram of dynamic networks of neurons for the removal of each single edge, obtained using the DNCMD algorithm. 
The labels represent the index of the removed edge (from $1$ to $13$) and are shown in different colors, as indicated in the legend on the right, and for each removed edge we have $11$ replicates.}
\label{fig:dendrogram}
\end{figure*}

Neuroscientifically, it is established (post facto) that edge $6$, which does form its own cluster, is in fact the edge that should be distinct --
it is the only edge associated with a feedback neuron (FBN). (Dopaminergic neuron DAN-f1 makes weak but reliable downstream connections onto FBN-1.)
This result provides a compelling proof of principle for the utility of DNCMD,
and motivates a follow-on analysis to provide predictions that can be tested experimentally in live animals.

In addition, we also develop numerical measures to quantify the concentration of replicates for each label and to assess functional similarity between edges in Section~\ref{sec:supp_brain} of the supplementary material. These quantitative analyses further confirm that edge $6$ is special and reveal that certain edge pairs (e.g., edges $1$ and $7$, edges $5$ and $9$) exhibit similar functions.

\subsection{Trade dynamic networks} \label{sec:trade}

We use the trade dynamic networks between countries for different
food and agriculture products during the year from $2005$ to $2022$. The data is
collected by the Food and Agriculture Organization of the United
Nations and is available at
\url{https://www.fao.org/faostat/en/#data/TM}.
We construct a collection of dynamic networks for $T=2022-2005+1=18$ time points, one dynamic network for each food or agriculture product,
where vertices represent countries
and edges in each network represent trade relationships between countries.
For each product $i$ and each time point $t$, we obtain the adjacency
matrix $\mathbf{A}^{(i)}_t$ by (1) we set $(\ma^{(i)}_t)_{r,s}=(\ma^{(i)}_t)_{s,r}=1$ if there is product $i$ trade between countries $r$ and $s$; (2) we ignore the links between
countries $r$ and $s$ in $\mathbf{A}^{(i)}_t$ if their
total trade amount for the product $i$ at the time point $t$ is less than two hundred thousands US
dollars; (3) finally we extract the {
 intersection} of the 
  largest connected components of $\{\mathbf{A}^{(i)}_t\}$ to get the networks for the common involved countries.
  The resulting adjacency matrices $\{\ma^{(i)}_t\}_{i\in[m],t\in[T]}$ corresponding to $m=18$ dynamic networks for $T=18$ time points on a set of $n=58$ vertices.

To analyze the relationships between trade patterns of these $m=18$ products, we apply Algorithm~\ref{Alg_clustering} to obtain the hierarchical clustering result for the $m=18$ dynamic networks with the embedding dimensions $d$ and $r$ chosen to be $2$. Figure~\ref{fig:food} presents the dendrogram of the hierarchical clustering result. 
From Figure~\ref{fig:food}, we see there is a high degree of correlation between the trade pattern relationships and the types of products. 
More specifically, beer and spirits both are alcoholic beverages produced from grain-based materials with similar consumer bases, and they show very similar trade patterns in the figure. Crude organic material n.e.c. and food preparations n.e.c. both have a certain ambiguity in their classification, and their trade trends are also be more challenging to match with other clearly classified products. Except the above four products, there appears to be two main clusters formed by raw/unprocessed products (top cluster) and processed products (bottom cluster), and it suggests discernable differences in the trading patterns for these types of products.
\begin{figure}[htbp] 
\centering
\includegraphics[width=7cm]{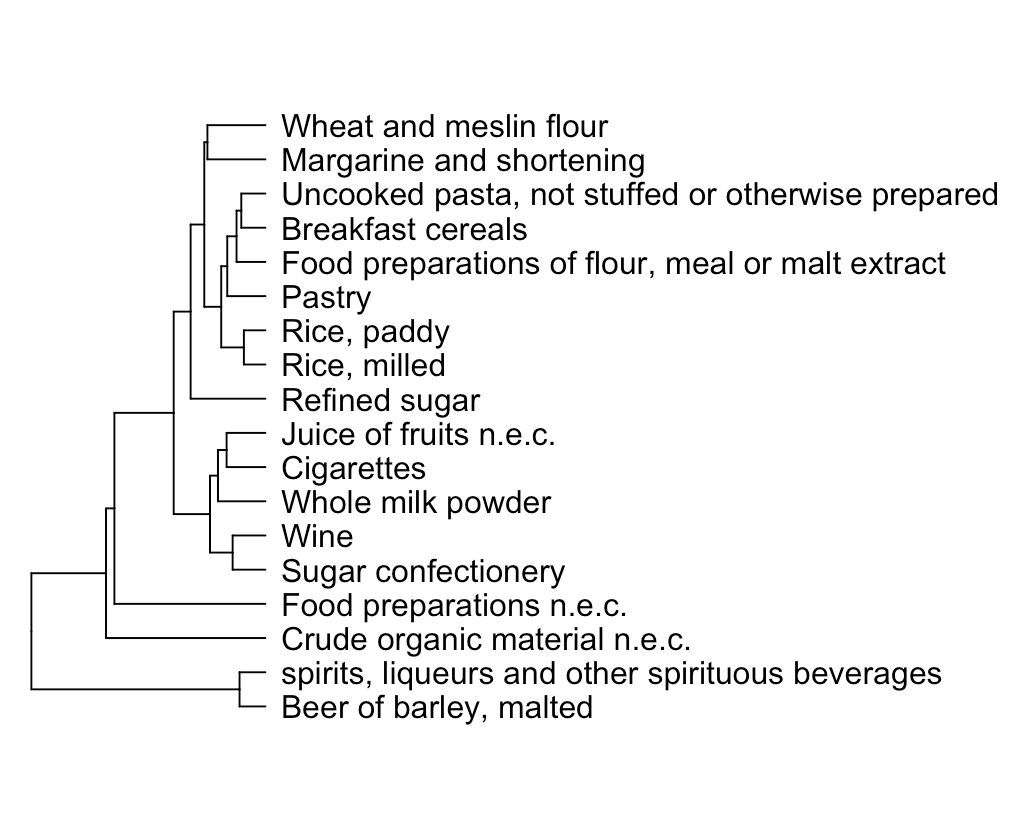}
\caption{Hierarchical clustering result for product trading patterns.}
\label{fig:food}
\end{figure}

\section{Conclusion and Discussion}

In summary, this paper presents DNCMD for clustering dynamic networks according to their evolution patterns.
Our algorithm is based on the mirror method, which captures the important features of the evolution of dynamic networks by curves in low-dimensional Euclidean space, and uses orthogonal Procrustes analysis to eliminate the non-identifiability of the mirrors, with the resulting pairwise distances used for hierarchical clustering to produce a dendrogram and the final clustering outcome.
Under mild assumptions, we establish theoretical guarantees for the exact recovery of our algorithm in two general scenarios: when vertex latent positions are deterministic and when vertex latent positions are randomly generated from stochastic processes.
The simulation experiments validate our theoretical results and demonstrate that our algorithm generally provides highly accurate clustering results.
We finally demonstrate the applications of our algorithm on real-world data. For the Drosophila larval connectome data, we show how to classify the dynamic neural networks obtained after removing each specific edge and leverage this result to indirectly analyze the function of edges in specific neural processes, such as learning a new association between an odor stimulus and a reward.
And we cluster trade dynamic networks to analyze the relationships in the evolution of food products. 

We now mention potential directions for future work. Our clustering approach is based on mirrors, and the mirror idea can be extended to other network models beyond RDPG. 
For instance, when the common subspace independent edge (COSIE) model \cite{arroyo2021inference} is applied to a dynamic network, the probability matrix at time point $t$ can be written as $\mpp_t = \muu\mr_t\muu^\top$, where $\muu$ represents the invariant latent structure, while the possibly time-varying connectivity pattern is modeled by $\{\mr_t\}_t$. Following the mirror idea, we can measure the pairwise differences between $\{\mr_t\}_t$, which represent the evolving underlying structure, to find the corresponding curve of the dynamic network in low-dimensional space under this model.
A feasible measure for the difference between $\hat\mr_{t_1}$ and $\hat\mr_{t_2}$ for any $t_1$ and $t_2$ can be found in Theorem~6 of \cite{zheng2022limit}.
Finally, the mirror approach, which finds low-dimensional configurations to extract main signals, along with the corresponding clustering method based on distances between mirrors, can be applied to problems beyond dynamic networks, as long as the problem involves similar grouped structures.

\bibliographystyle{chicago}
\bibliography{references}

\newpage

\newtheorem{lemma}{Lemma}[section]

\begin{center}%
    {\LARGE Supplementary Material for ``Dynamic networks clustering via mirror distance"\par}%
  \end{center}

\appendix

\setcounter{figure}{0}
\setcounter{table}{0}
\renewcommand{\thefigure}{\thesection.\arabic{figure}}
\renewcommand{\thetable}{\thesection.\arabic{table}}

\section{Additional Numerical Results}

\subsection{Additional analysis for Drosophila larval connectome}\label{sec:supp_brain}
Based on the dendrogram in Figure~\ref{fig:dendrogram}, we develop a numerical measure to quantify the degree of concentration of replicates for each label.
Given a number of clusters $K$, we can obtain a clustering result from the dendrogram and use a contingency table to display the frequency distribution of the clusters across the labels; see Panel~(a) of Figure~\ref{fig:contingency} for the contingency table when $K=3$ is set.
We denote the frequencies in the contingency table for $K$ clusters by $\{\mc_{i,k}^{(K)}\}_{i\in[13],k\in[K]}$, where $i$ is the index over the $13$ labels, $k$ is the index over the total $K$ clusters.
Since each label has $11$ observations, we compute the frequency rate among all dynamic networks with label $i$ as $\mr_{i,k}^{(K)}:=\mc_{i,k}^{(K)}/11$; see Panel~(b) of Figure~\ref{fig:contingency} for an example table of the corresponding frequency rates.
For any fixed $K$, if a label $i$ has a high maximum frequency rate $\max_{k\in[K]}\{\mr_{i,k}^{(K)}\}$, it suggests that the dynamic networks for this label are concentrated in one of the clusters.
Instead of using the maximum frequency rates for a specific $K$, we consider integrating them over a reasonable range of $K$.
For example, we tune $K$ from $1$ to  $K_{\max}=10$  to obtain the curve of the maximum frequency rate for each label, as shown in Panel (a) of Figure~\ref{fig:AUC}. Then, we compute the area under the curve (AUC) for each label, and normalize it by dividing the AUC by $(K_{\max}-1)$, scaling it between $0$ and $1$; see Panel~(b) of Figure~\ref{fig:AUC} for the normalized AUCs for all labels. Based on the normalized AUCs, we also see that the replicates for edge $6$ are the most concentrated. Figures~\ref{fig:contingency} and \ref{fig:AUC} quantitatively support our claim that edge $6$ is special.

\begin{figure}[htbp] 
\centering
\subfigure[Contingency table]{%
\includegraphics[height=4.5cm]{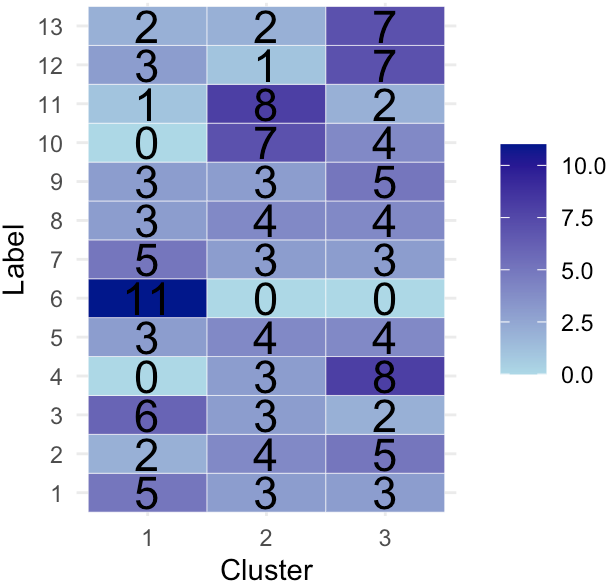}%
\label{fig:contingency_a}}
\hfil
\subfigure[Frequency rates]{%
\includegraphics[height=4.5cm]{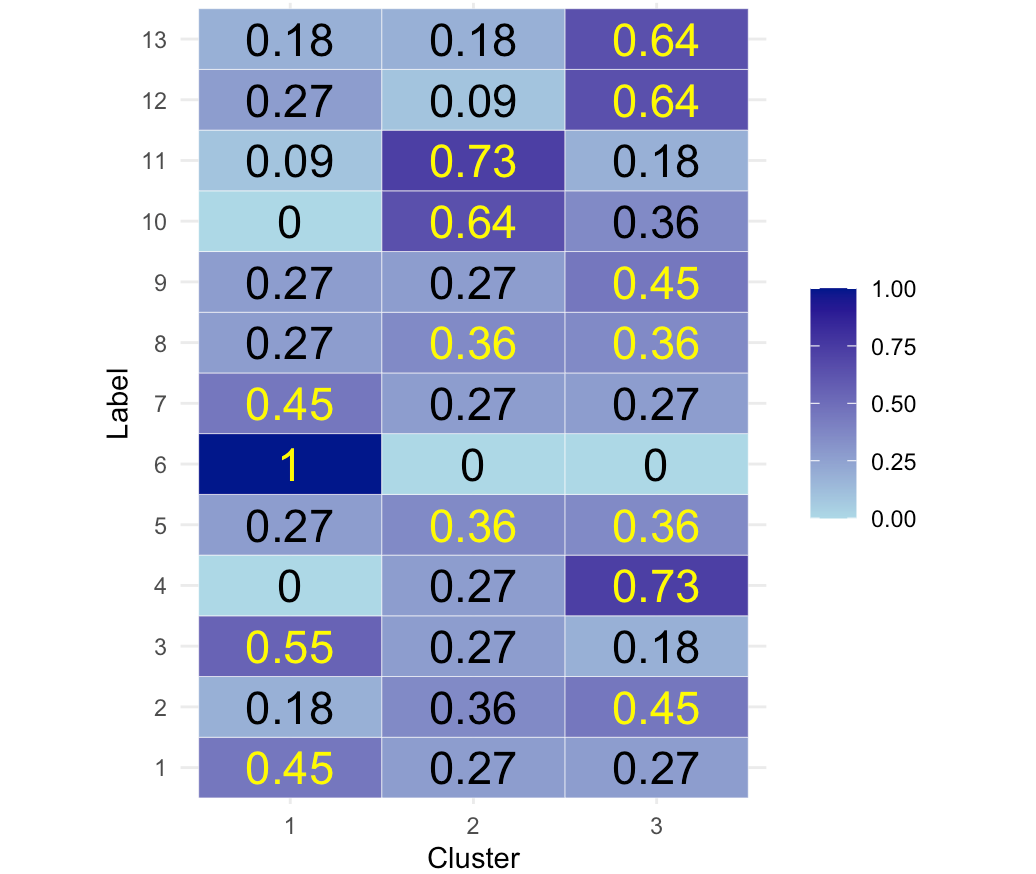}%
\label{fig:contingency_b}}
\caption{Panel (a) shows the contingency table for the case of $K=3$ clusters.
Panel (b) shows the corresponding frequency rates, and the maximum frequency rate for each label is highlighted in yellow.}
\label{fig:contingency}
\end{figure}

\begin{figure}[htbp] 
\centering
\subfigure[Maximum frequency rates]{%
\includegraphics[width=6cm]{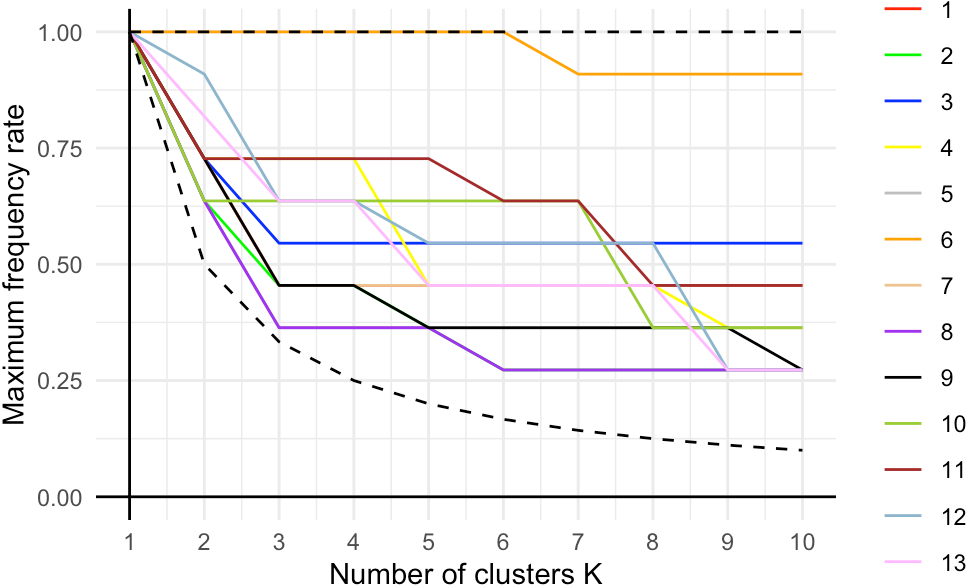}%
\label{fig:AUC_a}}
\subfigure[Normalized AUCs]{%
\includegraphics[width=6cm]{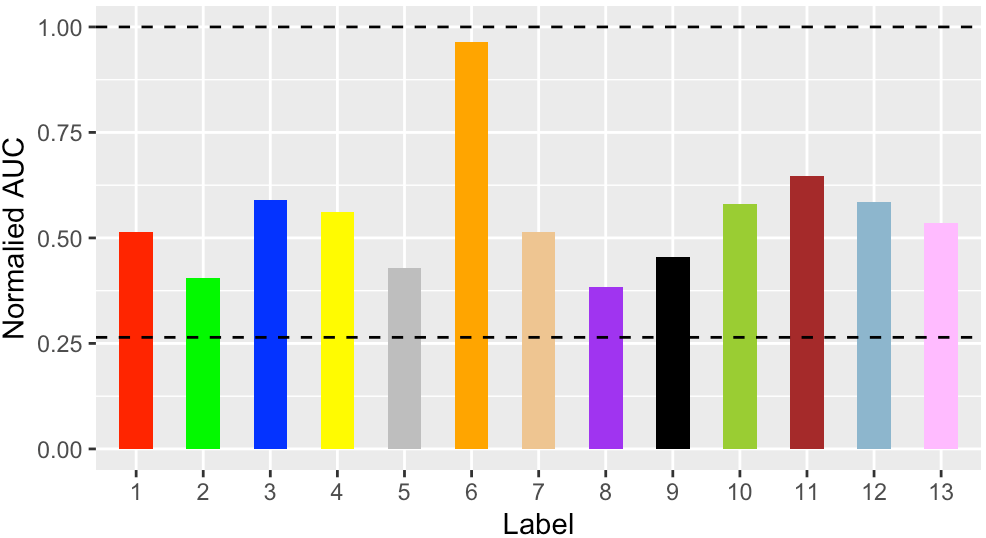}%
\label{fig:AUC_b}}
\caption{Panel~(a) shows the maximum frequency rates for all labels when 
we tune $K$ from $1$ to $K_{\max}=10$. In panel (a), the two dashed curves respectively describe the maximum value, $1$, and the minimum value, $1/K$, that the maximum frequency rate can achieve.
Panel~(b) shows the normalized AUCs for all labels, and the two dashed lines respectively describe the maximum value and the minimum value that the normalized AUC can achieve.}
\label{fig:AUC}
\end{figure}

One more question we are also concerned with is which edges might have similar functions. To address this, we measure the similarity between the distributions of replicates across clusters for each pair of edges.
Given a number of clusters $K$, for each pair of labels, we construct a measure of similarity using the Jaccard distance between their distributions across the clusters. The Jaccard distance ranges from $0$ to $1$, where $0$ indicates that the two distributions are exactly the same.
Thus, a small Jaccard distance suggests that the dynamic networks for these two labels have similar distributions across the clusters.
Figure~\ref{fig:Jaccard} shows the Jaccard distances for all pairs of labels when $K=3$.
We also consider integrating the Jaccard distances over a range of $K$.
We show the Jaccard distances for all pairs of labels when we tune $K$ from $1$ to $K_{\max}=10$ in Panel~(a) of Figure~\ref{fig:Jaccard2}. We then obtain their AUCs and also normalize it by dividing the AUC by $/(K_{\max}-1)$ in Panel~(b) of Figure~\ref{fig:Jaccard2}.
A pair of labels $(i,j)$ with a small normalized AUC of the Jaccard distance suggests that this pair of labels always has similar distributions across the clusters, further implying that the dynamic networks corresponding to these two edges have similar distributions.
From Panel~(b) of Figure~\ref{fig:Jaccard2}, we observe that edges $1$ and $7$, as well as edges $5$ and $9$, exhibit the most similar functions, whereas edge $6$ is distinct from all other edges.


\begin{figure}[htbp] 
\centering
\includegraphics[height=5cm]{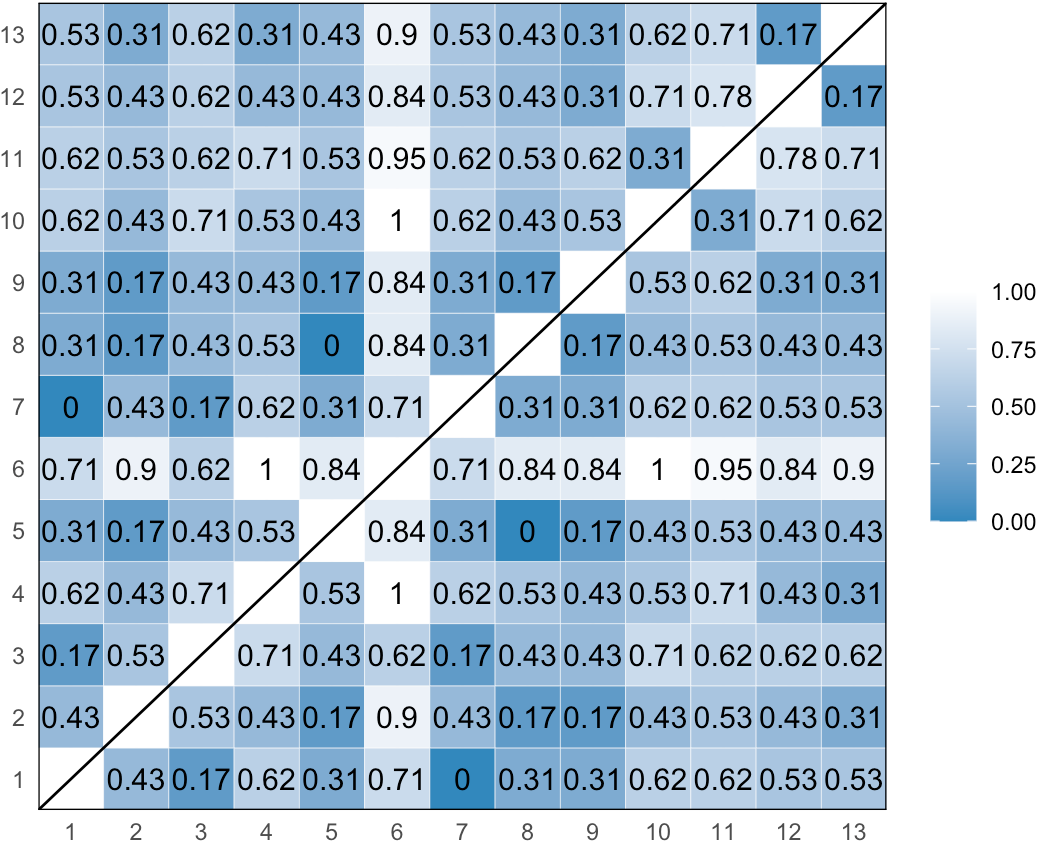}
\caption{Jaccard distances between the distributions across the clusters for all pairs of labels for $K=3$.}
\label{fig:Jaccard}
\end{figure}

\begin{figure}[htbp] 
\centering
\subfigure[Jaccard distances]{%
\includegraphics[height=5cm]{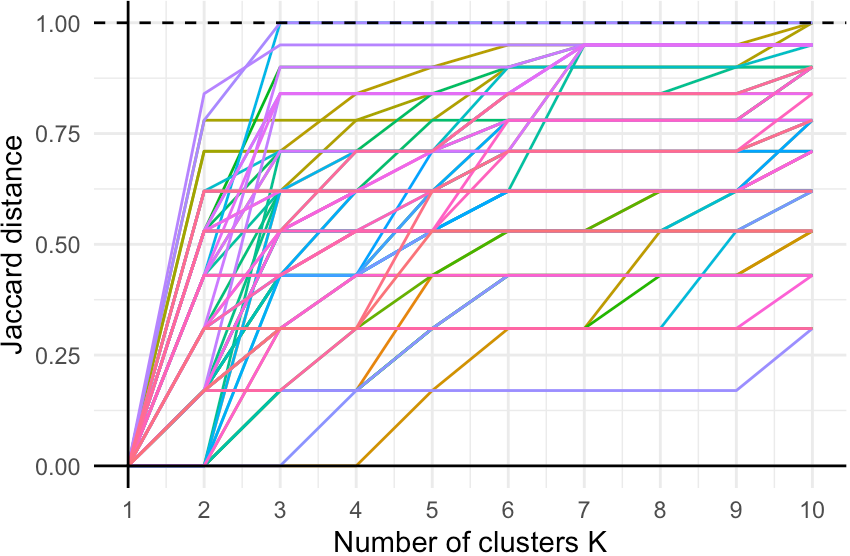}%
\label{fig:Jaccard2_a}}
\subfigure[Normalized AUCs]{%
\includegraphics[height=5cm]{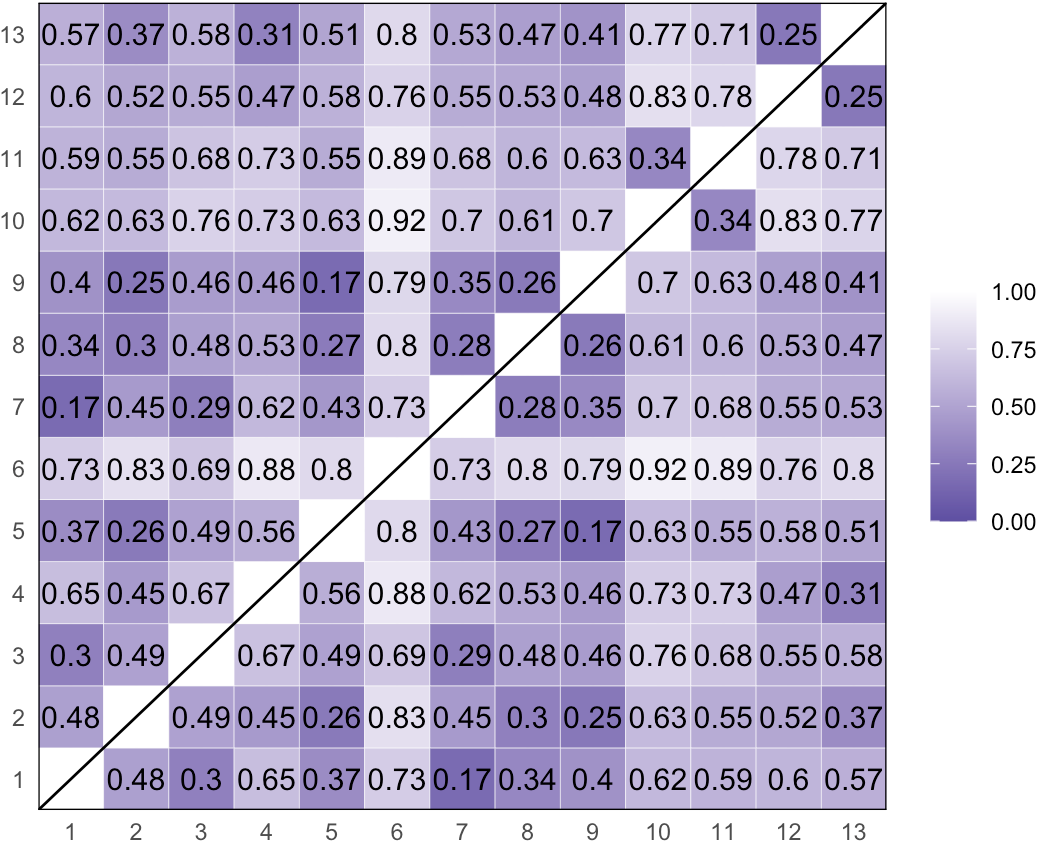}%
\label{fig:Jaccard2_b}}
\caption{Panel~(a) shows the Jaccard distances for all pairs of labels when 
we tune $K$ from $1$ to $K_{\max}=10$. 
The dashed line represents the maximum value, $1$, that the Jaccard distance can achieve. We have total $\binom{13}{2}=78$ lines in Panel~(a) to represent all pairs of these $13$ labels in different colors.
Panel~(b) shows the normalized AUCs for all pairs of labels.}
\label{fig:Jaccard2}
\end{figure}

\subsection{Additional results for trade dynamic networks}\label{sec:supp_trade}

To complement the hierarchical clustering result in Figure~\ref{fig:food} of the main text, we also apply CMDS to visualize the distance relationships between products based on $\md^{\star}$, as shown in Figure~\ref{fig:food_embedding}.

\begin{figure}[htbp] 
\centering
\includegraphics[width=9cm]{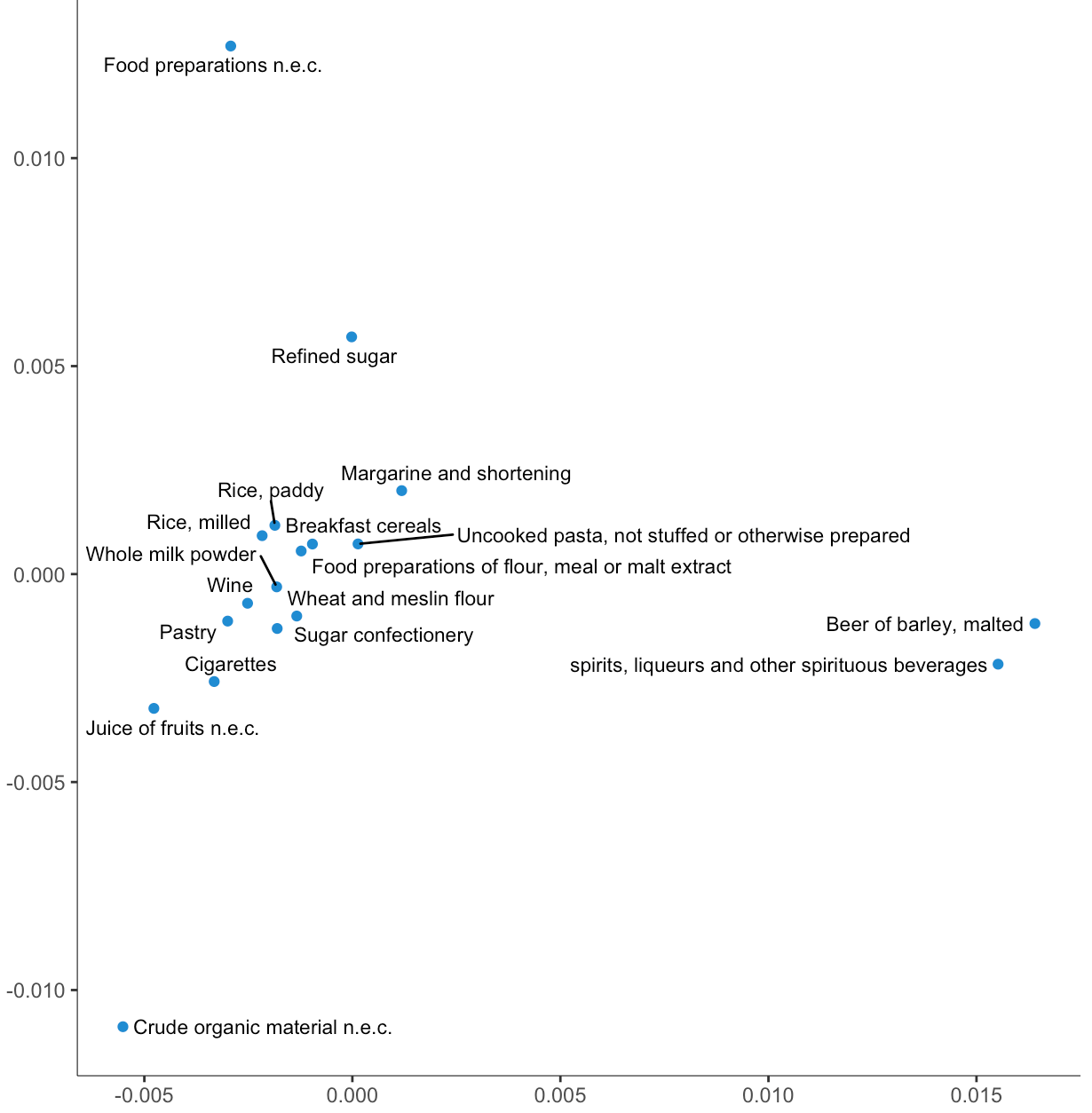}
\caption{$2$-dimensional embeddings based on $\md^\star$ for product trading patterns.}
\label{fig:food_embedding}
\end{figure}

\section{Proof of Main Results}

\subsection{Proof of Theorem~\ref{thm:same D}}
Because the dynamic networks $i,j$ are the same, for any $t\in[T]$ we have $\mx^{(i)}_t\mx^{(i)\top}_t=\mpp^{(i)}_t=\mpp^{(j)}_t=\mx^{(j)}_t\mx^{(j)\top}_t$, so there exists $\mw^{(i,j;t)}\in\mathcal{O}_d$ such that $\mx^{(i)}_t\mw^{(i,j;t)}=\mx^{(j)}_t$.

For any $t_1,t_2\in{T}$, we first prove $\md^{(j)}_{t_1,t_2}\geq\md^{(j)}_{t_1,t_2}$.
Let $\tilde\mw^{(j;t_1,t_2)}:=\underset{\mo\in\mathcal{O}_d}{\operatorname{argmin}}\|\mx^{(j)}_{t_1}\mo-\mx^{(j)}_{t_2}\|_F$.
Then we have
$$
\begin{aligned}
	\md^{(j)}_{t_1,t_2}
    &=n^{-1/2}\|\mx^{(j)}_{t_1}\tilde\mw^{(j;t_1,t_2)}-\mx^{(j)}_{t_2}\|_F\\
    &=n^{-1/2}\|\mx^{(i)}_{t_1}\mw^{(i,j;t_1)}\tilde\mw^{(j;t_1,t_2)}-\mx^{(i)}_{t_2}\mw^{(i,j;t_2)}\|_F\\
    &=n^{-1/2}\|\mx^{(i)}_{t_1}\mw^{(i,j;t_1)}\tilde\mw^{(j;t_1,t_2)}\mw^{(i,j;t_2)\top}-\mx^{(i)}_{t_2}\|_F
    \geq n^{-1/2}\min_{\mo\in\mathcal{O}_d}\|\mx^{(i)}_{t_1}\mo-\mx^{(i)}_{t_2}\|_F
    =\md^{(i)}_{t_1,t_2}
\end{aligned}
$$
because $\mw^{(i,j;t_1)}\tilde\mw^{(j;t_1,t_2)}\mw^{(i,j;t_2)\top}\in\mathcal{O}_d$.
With the identical analysis, by $\mx^{(i)}_t=\mx^{(j)}_t\mw^{(i,j;t)\top}$ we also have $\md^{(j)}_{t_1,t_2}\leq\md^{(j)}_{t_1,t_2}$. So finally we have $\md^{(j)}_{t_1,t_2}=\md^{(j)}_{t_1,t_2}$ for any $t_1,t_2\in{T}$, and therefore we have $\md^{(i)}=\md^{(j)}$.
\hspace*{\fill} $\square$

\subsection{Proof of Theorem~\ref{thm:same M}}

Because $\md^{(i)}=\md^{(j)}$, we have $\mb^{(i)}=\mb^{(j)}$. 
Consider two eigen-decompositions of $\mb^{(i)}$ that may be different $$\mb^{(i)}=\muu^{(i)}\mLambda^{(i)}\muu^{(i)\top}+\muu^{(i)}_\perp\mLambda^{(i)}_\perp\muu^{(i)\top}_\perp=\tilde\muu^{(i)}\mLambda^{(i)}\tilde\muu^{(i)\top}+\tilde\muu^{(i)}_\perp\mLambda^{(i)}_\perp\tilde\muu^{(i)\top}_\perp,$$ and let $\mm^{(i)}=\muu^{(i)}(\mLambda^{(i)})^{1/2}$, $\mm^{(j)}=\tilde\muu^{(i)}(\mLambda^{(i)})^{1/2}$.
Because $\lambda_r(\mb^{(i)})>\lambda_{r+1}(\mb^{(i)})$, $\mLambda^{(i)}$ contains different eigenvalues with $\mLambda^{(i)}_\perp$. Therefore $\text{span}(\muu^{(i)})$ and $\text{span}(\tilde\muu^{(i)})$ are the direct sum of the eigenspaces for all eigenvalues in $\mLambda^{(i)}$, where $\text{span}(\muu)$ denote the span of columns of $\muu$. It means that $\text{span}(\muu^{(i)})$ and $\text{span}(\tilde\muu^{(i)})$ must be the same, so their projection matrices are the same, i.e. $$\muu^{(i)}(\muu^{(i)\top}\muu^{(i)})^{-1}\muu^{(i)\top}=\tilde\muu^{(i)}(\tilde\muu^{(i)\top}\tilde\muu^{(i)})^{-1}\tilde\muu^{(i)\top}.$$
It follows that
$$
\muu^{(i)}\mLambda^{(i)}\muu^{(i)\top}
=\muu^{(i)}(\muu^{(i)\top}\muu^{(i)})^{-1}\muu^{(i)\top}\mb^{(i)}
=\tilde\muu^{(i)}(\tilde\muu^{(i)\top}\tilde\muu^{(i)})^{-1}\tilde\muu^{(i)\top}\mb^{(i)}
=\tilde\muu^{(i)}\mLambda^{(i)}\tilde\muu^{(i)\top}.
$$
Therefore $\mm^{(i)}\mm^{(i)\top}=\mm^{(j)}\mm^{(j)\top}$ and there exists $\mw_\mm^{(i,j)}\in\mathcal{O}_r$ such that $\mm^{(i)}\mw_\mm^{(i,j)}=\mm^{(j)}$.
\hspace*{\fill} $\square$

\subsection{Proof of Theorem~\ref{thm:D error}}

For ease of exposition we will fix a value of $i\in[m]$ and thereby drop the index $i$ from our matrices.

For any fixed $t_1,t_2\in[T]$, for the latent position $\mx_{t_1}$, $\mx_{t_2}$ and their corresponding estimates $\hat\mx_{t_1}$, $\hat\mx_{t_2}$ we define the following orthogonal matrices
$$
\begin{aligned}
	&\mw:=\underset{\mo\in\mathcal{O}_d}{\operatorname{argmin}}\|\mx_{t_1}\mo-\mx_{t_2}\|_F,\quad
	\hat\mw:=\underset{\mo\in\mathcal{O}_d}{\operatorname{argmin}}\|\hat\mx_{t_1}\mo-\hat\mx_{t_2}\|_F,\\
	&\mw^{(t_1)}:=\underset{\mo\in\mathcal{O}_d}{\operatorname{argmin}}\|\hat\mx_{t_1}\mo-\mx_{t_1}\|_F,\quad
	\mw^{(t_2)}:=\underset{\mo\in\mathcal{O}_d}{\operatorname{argmin}}\|\hat\mx_{t_2}\mo-\mx_{t_2}\|_F.
\end{aligned}
$$
Then for $|\hat\md_{t_1,t_2}-\md_{t_1,t_2}|$ we have
\begin{equation}\label{eq:D-D=...}
	\begin{aligned}
	n^{1/2} |\hat\md_{t_1,t_2}-\md_{t_1,t_2}|
	&= \big|\|\hat\mx_{t_1}\hat\mw-\hat\mx_{t_2}\|_F
	-\|\mx_{t_1}\mw-\mx_{t_2}\|_F
	\big|\\
	&=\big|\|\hat\mx_{t_1}\hat\mw\mw^{(t_2)}-\hat\mx_{t_2}\mw^{(t_2)}\|_F
	-\|\mx_{t_1}\mw-\mx_{t_2}\|_F\big|\\
	&\leq \|\hat\mx_{t_2}\mw^{(t_2)}-\mx_{t_2}\|_F
	+\|\hat\mx_{t_1}\hat\mw\mw^{(t_2)}-\mx_{t_1}\mw\|_F\\
	&\leq \|\hat\mx_{t_2}\mw^{(t_2)}-\mx_{t_2}\|_F
	+\|\hat\mx_{t_1}\mw^{(t_1)}-\mx_{t_1}\mw\mw^{(t_2)\top}\hat\mw^{\top}\mw^{(t_1)}\|_F\\
	&\leq \|\hat\mx_{t_2}\mw^{(t_2)}-\mx_{t_2}\|_F
	+\|\hat\mx_{t_1}\mw^{(t_1)}-\mx_{t_1}+\mx_{t_1}(\mi-\mw\mw^{(t_2)\top}\hat\mw^{\top}\mw^{(t_1)})\|_F\\
	&\leq \|\hat\mx_{t_2}\mw^{(t_2)}-\mx_{t_2}\|_F
	+\|\hat\mx_{t_1}\mw^{(t_1)}-\mx_{t_1}\|_F
	+\|\mx_{t_1}\|_F\cdot\|\mi-\mw\mw^{(t_2)\top}\hat\mw^{\top}\mw^{(t_1)}\|\\
	&\leq \|\hat\mx_{t_2}\mw^{(t_2)}-\mx_{t_2}\|_F
	+\|\hat\mx_{t_1}\mw^{(t_1)}-\mx_{t_1}\|_F
	+\|\mx_{t_1}\|_F\cdot\|\mw-\mw^{(t_1)\top}\hat\mw\mw^{(t_2)}\|.
\end{aligned}
\end{equation}
For $\hat\mx_{t_1}\mw^{(t_1)}-\mx_{t_1}$ and $\hat\mx_{t_2}\mw^{(t_2)}-\mx_{t_2}$, according to Theorem~2.1 in \cite{tang2018limit} {\color{black}under the assumption $n\rho_n =\Omega(\log n)$} 
we have
\begin{equation}\label{eq:||X-XW||}
\begin{aligned}
	\|\hat\mx_{t_1}\mw^{(t_1)}-\mx_{t_1}\|_F\lesssim 1
	\text{ and }
	\|\hat\mx_{t_2}\mw^{(t_2)}-\mx_{t_2}\|_F\lesssim 1
\end{aligned}
\end{equation}
with high probability.
For $\mx_{t_1}$ and $\mx_{t_2}$, because $\|\mpp_{t_1}\|\asymp n\rho_n$ and $\|\mpp_{t_2}\|\asymp n\rho_n$ we have 
\begin{equation}\label{eq:||X||}
	\begin{aligned}
		\|\mx_{t_1}\|\asymp  (n\rho_n)^{1/2}
\text{ and }\|\mx_{t_2}\|\asymp  (n\rho_n)^{1/2}.
	\end{aligned}	
\end{equation}
We now bound $\mw-\mw^{(t_1)\top}\hat\mw\mw^{(t_2)}$. Notice
$$
\begin{aligned}
&\mw^{(t_1)\top}\hat\mw\mw^{(t_2)}
=\underset{\mo\in\mathcal{O}_d}{\operatorname{argmin}}\|\hat\mx_{t_1}\mw^{(t_1)}\mo-\hat\mx_{t_2}\mw^{(t_2)}\|_F,\\
&\mw=\underset{\mo\in\mathcal{O}_d}{\operatorname{argmin}}\|\mx_{t_1}\mo-\mx_{t_2}\|_F.
\end{aligned}
$$
We therefore have, by perturbation bounds for polar decompositions, that
\begin{equation}
 \label{eq:rencang} 
  \|\mw-\mw^{(t_1)\top}\hat\mw\mw^{(t_2)}\|
\leq \frac{2\|\mw^{(t_1)\top}\hat\mx_{t_1}^\top\hat\mx_{t_2}\mw^{(t_2)}-\mx_{t_1}^\top\mx_{t_2}\|}
{\sigma_{\min}(\mx_{t_1}^\top\mx_{t_2})}. 
\end{equation}
Indeed, $\mx_{t_1}^\top\mx_{t_2}$ is invertible.
Now suppose $\|\mw^{(t_1)\top}\hat\mx_{t_1}^\top\hat\mx_{t_2}\mw^{(t_2)}-\mx_{t_1}^\top\mx_{t_2}\| <  \sigma_{\min}(\mx_{t_1}^\top\mx_{t_2})$. Then 
$\hat\mx_{t_1}^\top\hat\mx_{t_2}$ is also invertible and hence, by Theorem~1 in \cite{rencang} we have
$$
\begin{aligned}
	 \|\mw-\mw^{(t_1)\top}\hat\mw\mw^{(t_2)}\| &\leq \frac{2\|\mw^{(t_1)\top}\hat\mx_{t_1}^\top\hat\mx_{t_2}\mw^{(t_2)}-\mx_{t_1}^\top\mx_{t_2}\|}{\sigma_{\min}(\hat\mx_{t_1}^\top\hat\mx_{t_2}) + \sigma_{\min}(\mx_{t_1}^\top\mx_{t_2})} \leq \frac{2\|\mw^{(t_1)\top}\hat\mx_{t_1}^\top\hat\mx_{t_2}\mw^{(t_2)}-\mx_{t_1}^\top\mx_{t_2}\|}{\sigma_{\min}(\mx_{t_1}^\top\mx_{t_2})}.
\end{aligned}
$$
Otherwise if $\|\mw^{(t_1)\top}\hat\mx_{t_1}^\top\hat\mx_{t_2}\mw^{(t_2)}-\mx_{t_1}^\top\mx_{t_2}\| \geq  \sigma_{\min}(\mx_{t_1}^\top\mx_{t_2})$
then, as $\|\mw-\mw^{(t_1)\top}\hat\mw\mw^{(t_2)}\| \leq 2$, Eq.~\eqref{eq:rencang} holds trivially. 
By Eq.~\eqref{eq:||X-XW||} and Eq.~\eqref{eq:||X||} we have
$$
\begin{aligned}
	\|\mw^{(t_1)\top}\hat\mx_{t_1}^\top\hat\mx_{t_2}\mw^{(t_2)}-\mx_{t_1}^\top\mx_{t_2}\|
	=&\|(\hat\mx_{t_1}\mw^{(t_1)}-\mx_{t_1})^\top(\hat\mx_{t_2}\mw^{(t_2)}-\mx_{t_2})
	\\+&(\hat\mx_{t_1}\mw^{(t_1)}-\mx_{t_1})^\top\mx_{t_2}
	+\mx_{t_1}^\top(\hat\mx_{t_2}\mw^{(t_2)}-\mx_{t_2})\|\\
	\leq &\|\hat\mx_{t_1}\mw^{(t_1)}-\mx_{t_1}\|_F\cdot\|\hat\mx_{t_2}\mw^{(t_2)}-\mx_{t_2}\|_F
	\\+&\|\hat\mx_{t_1}\mw^{(t_1)}-\mx_{t_1}\|_F\cdot\|\mx_{t_2}\|
	+\|\mx_{t_1}\|\cdot\|\hat\mx_{t_2}\mw^{(t_2)}-\mx_{t_2}\|_F\\
	\lesssim &1\cdot 1+1\cdot  (n\rho_n)^{1/2}+  (n\rho_n)^{1/2}\cdot 1
	\lesssim  (n\rho_n)^{1/2}
\end{aligned}
$$
with high probability.
Notice we have $\lambda_k(\mpp_{t_1})\asymp n\rho_n$ and $\lambda_k(\mpp_{t_2})\asymp n\rho_n$ for any $k\in[d]$, it follows that $\sigma_k(\mx_{t_1}^\top\mx_{t_2})\asymp n\rho_n $ for any $k\in[d]$. Then by Eq.~\eqref{eq:rencang} we have
\begin{equation}\label{eq:|W-WWW|}
	\begin{aligned}
		\|\mw-\mw^{(t_1)\top}\hat\mw\mw^{(t_2)}\|
\lesssim \frac{(n\rho_n)^{1/2}}
{n \rho_n}
\lesssim  (n\rho_n)^{-1/2}
	\end{aligned}
\end{equation}
with high probability.
Combining Eq.~\eqref{eq:D-D=...}, Eq.~\eqref{eq:||X-XW||}, Eq.~\eqref{eq:||X||} and Eq.~\eqref{eq:|W-WWW|} we have
$$
 n^{1/2}|\hat\md_{t_1,t_2}-\md_{t_1,t_2}|
\lesssim 1+ 1+ (n\rho_n)^{1/2}\cdot (n\rho_n)^{-1/2}
\lesssim 1
$$
with high probability. Therefore the desired result for $\|\hat\md-\md\|_{\max}$ is obtained.
\hspace*{\fill} $\square$

\subsection{Proof of Theorem~\ref{thm:M error}}
For ease of exposition we will fix a value of $i\in[m]$ and thereby drop the index $i$ from our matrices.

For any orthogonal matrix $\mw$ we have
\begin{equation*}
	\begin{aligned}
		\hat\mm\mw-\mm
		&=\hat\muu\hat\mLambda^{1/2}\mw-\muu\mLambda^{1/2}\\
		&=\muu\mw^\top(\hat\mLambda^{1/2}\mw-\mw\mLambda^{1/2})+(\hat\muu\mw-\muu)\mw^\top\hat\mLambda^{1/2}\mw.
	\end{aligned}
\end{equation*}
Then by Lemma~\ref{lemma:Lambda1/2W-WLambda1/2} there exists orthogonal matrix $\mw$ such that
	$$
	\begin{aligned}
		\|\hat\mm\mw-\mm\|_{F}
		&\leq r^{1/2}[\|\hat\mLambda^{1/2}\mw-\mw\mLambda^{1/2}\|
		+\|\hat\muu\mw-\muu\|\cdot \|\hat\mLambda^{1/2}\|]\\
		&\lesssim \frac{T(n\rho_n)^{1/2}}{n \lambda_r^{1/2}(\mb^{(i)})}
		\Big(1+\frac{T(n\rho_n)^{1/2}\lambda_1(\mb)}{n \lambda_r^{2}(\mb^{(i)})}\Big)
		+\frac{T(n\rho_n)^{1/2}}{n \lambda_r(\mb^{(i)})}
		\cdot \lambda_1^{1/2}(\mb)\\
		&\lesssim \frac{T(n\rho_n)^{1/2}\lambda_1^{1/2}(\mb^{(i)})}{n \lambda_r(\mb^{(i)})}
		\Big(1+\frac{T(n\rho_n)^{1/2}\lambda_1^{1/2}(\mb)}{n \lambda_r^{3/2}(\mb^{(i)})}\Big)
	\end{aligned}
	$$
	with high probability, and the desired result of $\hat\mm$ is obtained.
\hspace*{\fill} $\square$

\begin{lemma}\label{lemma:Lambda1/2W-WLambda1/2}
Consider the setting of Theorem~\ref{thm:M error}.
Then for $\hat\mb^{(i)}$ we have
\begin{equation}\label{eq:lambda B}
	\begin{aligned}
	    &\lambda_k(\hat\mb^{(i)})\asymp \lambda_k(\mb^{(i)})\text{ for }k=1,\dots,r,\\&\lambda_{k}(\hat\mb^{(i)})\lesssim T n^{-1}(n\rho_n)^{1/2}\text{ for }k=r+1,\dots,T
	\end{aligned}
\end{equation}
	with high probability,
and for $\hat\mLambda^{(i)}$ and $\hat\muu^{(i)}$ there exists orthogonal matrix $\mw$ such that
$$
	\begin{aligned}
	    &\|\hat\muu^{(i)}\mw-\muu^{(i)}\|\lesssim \frac{T(n\rho_n)^{1/2}}{n \lambda_r(\mb^{(i)})},\\
		&\|(\hat\mLambda^{(i)})^{1/2}\mw-\mw(\mLambda^{(i)})^{1/2}\|
		\lesssim \frac{T(n\rho_n)^{1/2}}{n \lambda_r^{1/2}(\mb^{(i)})}
		\Big(1+\frac{T(n\rho_n)^{1/2}\lambda_1(\mb)}{n \lambda_r^{2}(\mb^{(i)})}\Big)
	\end{aligned}
	$$
	with high probability.
\end{lemma}

\begin{proof}
For ease of exposition we will fix a value of $i\in[m]$ and thereby drop the index $i$ from our matrices.

Let $\me_\mb:=\hat\mb-\mb$ and $\delta= T n^{-1}(n\rho_n)^{1/2}$. Then according to Lemma~\ref{lemma:||hatB-B||} we have $\|\me_\mb\|\leq \delta$ with high probability.
Then by perturbation theorem for singular values (see Problem~III.6.13 in \cite{horn2012matrix}) we have $$|\lambda_k(\hat\mb)-\lambda_k(\mb)|\leq \|\me_\mb\|\lesssim \delta$$
with high probability for all $k\in[T]$. Therefore {\color{black} under the assumption $\lambda_r(\mb)=\omega(\delta)$ and $\lambda_{r+1}(\mb)=O(\delta)$} we have the desired result of the eigenvalues of $\hat\mb^{(i)}$.

	By Eq.~\eqref{eq:lambda B} and Wedin's $\sin\Theta$ Theorem (see e.g., Theorem~4.4 in Chapter~4 of  \cite{stewart1990matrix}) we have
	$$
	\begin{aligned}
		\|\sin\Theta(\hat\muu,\muu)\|
		\leq \frac{\|\me_\mb\|}{\lambda_r(\hat\mb)-\lambda_{r+1}(\mb)}
		\lesssim \frac{\delta}{\lambda_r(\mb)}
	\end{aligned}
	$$
	with high probability.
	Therefore there exists orthogonal matrix $\mw$ such that
	\begin{equation}\label{eq:UU-W}
		\begin{aligned}
		&\|\muu^{\top}\hat\muu-\mw^{\top}\|
    	\leq \|\sin\Theta(\hat\muu,\muu)\|^2
    	\lesssim \frac{\delta^2}{\lambda_r^2(\mb)},\\
    	&\|\hat\muu\mw-\muu\|
    	\leq \|\sin\Theta(\hat\muu,\muu)\|
    	+\|\muu^{\top}\hat\muu-\mw^{\top}\|
    	\lesssim \frac{\delta}{\lambda_r(\mb)}
	\end{aligned}
	\end{equation}
	with high probability.
	
	For $\mLambda\muu^{\top}\hat\muu-\muu^{\top}\hat\muu\hat\mLambda$, we have
	\begin{equation}\label{eq:LambdaUU-UULambda}
		\begin{aligned}
		\|\mLambda\muu^{\top}\hat\muu-\muu^{\top}\hat\muu\hat\mLambda\|
		=\|-\muu^{\top}\me_\mb\hat\muu\|
		\leq\|\me_\mb\|
		\lesssim \delta
	\end{aligned}
	\end{equation}
	with high probability. By Eq.~\eqref{eq:lambda B}, Eq.~\eqref{eq:UU-W} and Eq.~\eqref{eq:LambdaUU-UULambda} we have
	\begin{equation}\label{eq:LambdaW-WLambda}
		\begin{aligned}
		\|\mLambda\mw^{\top}
		-\mw^{\top}\hat\mLambda\|
		&=\|\mLambda(\mw^{\top}-\muu^\top\hat\muu)
		+(\mLambda\muu^\top\hat\muu-\muu^\top\hat\muu\hat\mLambda)
		+(\muu^\top\hat\muu-\mw^{\top})\hat\mLambda\|\\
		&\leq \|\mLambda\muu^{\top}\hat\muu-\muu^{\top}\hat\muu\hat\mLambda\|
		+\|\muu^\top\hat\muu-\mw^{\top}\|\cdot (\|\mLambda\|+\|\hat\mLambda\|)\\
		&\lesssim \delta+\frac{\delta^2}{\lambda_r^2(\mb)}\cdot \lambda_1(\mb)
		\lesssim \delta+\frac{\delta^2\lambda_1(\mb)}{\lambda_r^2(\mb)}
	\end{aligned}
	\end{equation}
	with high probability.
	Then for $\hat\mLambda^{1/2}\mw-\mw\mLambda^{1/2}$, because
	$$
	\begin{aligned}
		\hat\mLambda^{1/2}\mw-\mw\mLambda^{1/2}
		=(\hat\mLambda\mw-\mw\mLambda)\circ \mh,
	\end{aligned}
	$$
	where we define $\mh$ as an $r\times r$ matrix with entries $\mh_{k,\ell}=(\sqrt{\lambda_\ell(\hat\mb)}+\sqrt{\lambda_k(\mb)})^{-1}$ for $k,\ell\in[r]$, and thus $\|\mh\|_{\max}\lesssim \lambda_r^{-1/2}(\mb)$ with high probability. Therefore by Eq.~\eqref{eq:LambdaW-WLambda} we have
	$$
	\begin{aligned}
		\|\hat\mLambda^{1/2}\mw-\mw\mLambda^{1/2}\|
		\leq r^{1/2}\|\hat\mLambda\mw-\mw\mLambda\|\cdot\|\mh\|_{\max}
		\lesssim \frac{\delta}{\lambda_r^{1/2}(\mb)}+\frac{\delta^2\lambda_1(\mb)}{\lambda_r^{5/2}(\mb)}
	\end{aligned}
	$$
	with high probability.
\end{proof}

\begin{lemma}\label{lemma:||hatB-B||}
Consider the setting of Theorem~\ref{thm:M error}.
For the error between $\hat\mb^{(i)}$ and $\mb^{(i)}$ we have
$$
\begin{aligned}
	&\|\hat\mb^{(i)}-\mb^{(i)}\|_{\max}\lesssim n^{-1}(n\rho_n)^{1/2},\\
	&\|\hat\mb^{(i)}-\mb^{(i)}\|\lesssim T n^{-1}(n\rho_n)^{1/2}
\end{aligned}
$$
with high probability.
\end{lemma}

\begin{proof}
For ease of exposition we will fix a value of $i\in[m]$ and thereby drop the index $i$ from our matrices.

	Let $\me_{\md^{\circ 2}}:=\hat\md^{\circ 2}-\md^{\circ 2}$.
	For any $t_1,t_2\in[T]$, we have 
	\begin{equation}\label{eq:Ett}
		\begin{aligned}
			|(\me_{\md^{\circ 2}})_{t_1,t_2}|=|\hat\md_{t_1,t_2}^2-\md_{t_1,t_2}^2|\leq|\hat\md_{t_1,t_2}-\md_{t_1,t_2}|\cdot |\hat\md_{t_1,t_2}+\md_{t_1,t_2}|.
		\end{aligned}
	\end{equation}
	By Theorem~\ref{thm:D error} we have 
	\begin{equation}\label{eq:Dtt-Dtt}
		\begin{aligned}
			|\hat\md_{t_1,t_2}-\md_{t_1,t_2}|\lesssim n^{-1/2}
		\end{aligned}
	\end{equation} with high probability. For $\md_{t_1,t_2}$, because
	$
	|\md_{t_1,t_2}|=n^{-1/2}\|\mx_{t_1}\mw-\mx_{t_2}\|_F
	$ for some $\mw\in\mathcal{O}_d$, by Eq.~\eqref{eq:||X||} we have
	\begin{equation}\label{eq:Dtt}
		\begin{aligned}
			|\md_{t_1,t_2}|\leq n^{-1/2}(\|\mx_{t_1}\|_F+\|\mx_{t_2}\|_F)
	\lesssim n^{-1/2}(n\rho_n)^{1/2}\lesssim \rho_n^{1/2}.
		\end{aligned}
	\end{equation}
	For $\hat\md_{t_1,t_2}$, by Eq.~\eqref{eq:Dtt-Dtt} and Eq.~\eqref{eq:Dtt} we have
	\begin{equation}\label{eq:hatDtt}
		\begin{aligned}
			|\hat\md_{t_1,t_2}|
			\leq |\md_{t_1,t_2}|+|\hat\md_{t_1,t_2}-\md_{t_1,t_2}|
			\lesssim \rho_n^{1/2}+n^{-1/2}
			\lesssim \rho_n^{1/2}
		\end{aligned}
	\end{equation}
	with high probability {\color{black} under the assumption $n\rho_n=\Omega(\log n)$}.
	Combining Eq.~\eqref{eq:Ett}, Eq.~\eqref{eq:Dtt-Dtt}, Eq.~\eqref{eq:Dtt} and Eq.~\eqref{eq:hatDtt} we have
	\begin{equation*}
	\begin{aligned}
		|(\me_{\md^{\circ 2}})_{t_1,t_2}|
		\lesssim n^{-1/2}\cdot \rho_n^{1/2}
		\lesssim n^{-1}(n\rho_n)^{1/2}
	\end{aligned}	
	\end{equation*}
	with high probability. Therefore we have
	\begin{equation}\label{eq:E}
		\begin{aligned}
			\|\me_{\md^{\circ 2}}\|_{\max} 
			\lesssim n^{-1}(n\rho_n)^{1/2}
		\end{aligned}
	\end{equation}
    with high probability.
	For the centering matrix $\mj$, it is easy to know $\|\mj\|_1\lesssim 1$ and $\|\mj\|_\infty\lesssim 1$.
		
	$$
	\begin{aligned}
		\|\hat\mb-\mb\|_{\max}
		=\frac{1}{2}\|\mj\me_{\md^{\circ 2}}\mj\|_{\max}
		\leq \frac{1}{2}\|\mj\|_1\cdot \|\me_{\md^{\circ 2}}\|_{\max}\cdot \|\mj\|_\infty
		\lesssim  1\cdot n^{-1}(n\rho_n)^{1/2}\cdot 1
		\lesssim n^{-1}(n\rho_n)^{1/2}
	\end{aligned}
	$$
	with high probability. It follows that
	$
		\|\hat\mb-\mb\|
		\lesssim T n^{-1}(n\rho_n)^{1/2}
	$
	with high probability.
\end{proof}

\subsection{Proof of Theorem~\ref{thm:Dstar error}}

For any dynamic networks $i,j\in[m]$, by Theorem~\ref{thm:M error} there exist $\mw_{\mm}^{(i)}$ and $\mw_{\mm}^{(j)}\in\mathcal{O}_r$ such that
\begin{equation}\label{eq:hatMW-M}
	\begin{aligned}
	\frac{1}{\sqrt{T}}\|\hat\mm^{(i)}\mw_{\mm}^{(i)}-\mm^{(i)}\|_F\lesssim
 \frac{T^{1/2}(n\rho_n)^{1/2}}{n \lambda_r^{1/2}(\mb^{(i)})},
\quad \frac{1}{\sqrt{T}}\|\hat\mm^{(j)}\mw_{\mm}^{(j)}-\mm^{(j)}\|_F\lesssim
 \frac{T^{1/2}(n\rho_n)^{1/2}}{n \lambda_r^{1/2}(\mb^{(j)})}
\end{aligned}
\end{equation}
with high probability. 

We now bound the error for $\hat\md^\star_{i,j}$ as an estimator of $\md^\star_{i,j}$.
We define the following orthogonal matrices
$$
\begin{aligned}
	&\mw_\mm^{(i,j)}:=\underset{\mo\in\mathcal{O}_r}{\operatorname{argmin}}\|\mm^{(i)}\mo-\mm^{(j)}\|_F,\quad
	\hat\mw_\mm^{(i,j)}:=\underset{\mo\in\mathcal{O}_r}{\operatorname{argmin}}\|\hat\mm^{(i)}\mo-\hat\mm^{(j)}\|_F.
\end{aligned}
$$
Then for $|\hat\md^\star_{i,j}-\md^\star_{i,j}|$, with the identical analysis of Eq.~\eqref{eq:D-D=...} we have
\begin{equation}\label{eq:D-D=...star}
	\begin{aligned}
	|\hat\md^\star_{i,j}-\md^\star_{i,j}|
	&\leq \frac{1}{\sqrt{T}}\|\hat\mm^{(i)}\mw_{\mm}^{(i)}-\mm^{(i)}\|_F
	+\frac{1}{\sqrt{T}}\|\hat\mm^{(j)}\mw_\mm^{(j)}-\mm^{(j)}\|_F\\
	&+\frac{1}{\sqrt{T}}\min\{\|\mm^{(i)}\|_F,\|\mm^{(j)}\|_F\}\cdot\|\mw_\mm^{(i,j)}-\mw_{\mm}^{(i)\top}\hat\mw_\mm^{(i,j)}\mw_\mm^{(j)}\|.
\end{aligned}
\end{equation}
Furthermore, with the identical analysis of Eq.~\eqref{eq:|W-WWW|} and by Eq.~\eqref{eq:hatMW-M} we have
\begin{equation}\label{eq:W-WWWstar}
	\begin{aligned}
		\|\mw_\mm^{(i,j)}-\mw_{\mm}^{(i)\top}\hat\mw_\mm^{(i,j)}\mw_\mm^{(j)}\|
		&\lesssim \big[\|\hat\mm^{(i)}\mw_{\mm}^{(i)}-\mm^{(i)}\|_F
		\cdot \|\hat\mm^{(j)}\mw_{\mm}^{(j)}-\mm^{(j)}\|_F+\|\hat\mm^{(i)}\mw_{\mm}^{(i)}-\mm^{(i)}\|_F
		\cdot \|\mm^{(j)}\|_F\\
		&
		+\|\mm^{(i)}\|_F
		\cdot \|\hat\mm^{(j)}\mw_{\mm}^{(j)}-\mm^{(j)}\|_F\big]
		/[\sigma_{r-1}(\mm^{(i)\top}\mm^{(j)})
		+\sigma_{r}(\mm^{(i)\top}\mm^{(j)})]
		\\
		&\lesssim \Big[\frac{T(n\rho_n)^{1/2}}{n \lambda_r^{1/2}(\mb^{(i)})}
		\cdot \frac{T(n\rho_n)^{1/2}}{n \lambda_r^{1/2}(\mb^{(j)})}
		+\frac{T(n\rho_n)^{1/2}}{n \lambda_r^{1/2}(\mb^{(i)})}\cdot \lambda_r^{1/2}(\mb^{(j)})\\
		&
		+\lambda_r^{1/2}(\mb^{(i)})\cdot \frac{T(n\rho_n)^{1/2}}{n \lambda_r^{1/2}(\mb^{(j)})}\Big]
		\cdot [\lambda_r^{1/2}(\mb^{(i)})\lambda_r^{1/2}(\mb^{(j)})]^{-1}\\
		&\lesssim \frac{T(n\rho_n)^{1/2}}{n \lambda}
	\end{aligned}
\end{equation}
with high probability {\color{black}under the assumption $\frac{T(n\rho_n)^{1/2}}{n  \lambda_r}=O(1)$}.
Therefore by Eq.~\eqref{eq:hatMW-M}, Eq.~\eqref{eq:D-D=...star} and Eq.~\eqref{eq:W-WWWstar} we have
$$
\begin{aligned}
	|\hat\md^\star_{i,j}-\md^\star_{i,j}|
	&\lesssim \frac{T^{1/2}(n\rho_n)^{1/2}}{n \lambda_r^{1/2}(\mb^{(i)})}
	+\frac{T^{1/2}(n\rho_n)^{1/2}}{n \lambda_r^{1/2}(\mb^{(j)})}
	+\frac{1}{\sqrt{T}}\cdot \min\{\lambda_r^{1/2}(\mb^{(i)}),\lambda_r^{1/2}(\mb^{(j)})\}\cdot \frac{T(n\rho_n)^{1/2}}{n\lambda}\\
	&\lesssim \frac{T^{1/2}(n\rho_n)^{1/2}}{n \lambda^{1/2}}
\end{aligned}
$$
with high probability. The desired result of $\hat\md^{\star}$ follows.
\hspace*{\fill} $\square$

\subsection{Proof of Proposition~\ref{prop:M error2}}
	
	Recall that $\{\mathbf{m}^{(i)}_t\}_{t\in[T]}$ approximate the change of $\{\mx_t^{(i)}\}_{t\in[T]}$ we have
	$$\|\mathbf{m}^{(i)}_{t_1}-\mathbf{m}^{(i)}_{t_2}\|^2\asymp (\md^{(i)}_{t_1,t_2})^2=\frac{1}{n}\min_{\mo\in\mathcal{O}_d}\|\mx^{(i)}_{t_1}\mo-\mx^{(i)}_{t_2}\|_F^2\asymp \rho_n \tilde \rho \quad\text{for }t_1\neq t_2.$$
	It follows that
	$$\|\mathbf{m}^{(i)}_{t}\|^2
	=\Big\|\mathbf{m}^{(i)}_{t}-\frac{1}{T}\sum_{t'\in[T]} \mathbf{m}^{(i)}_{t'}\Big\|^2
	\leq \frac{1}{T}\sum_{t'\in[T]} \|\mathbf{m}^{(i)}_{t}-\mathbf{m}^{(i)}_{t'}\|^2
	\asymp \rho_n \tilde \rho
	\quad 
	\text{for each }t\in[T], 
	$$
	because $\{\mathbf{m}^{(i)}_t\}_{t\in[T]}$ are centered, i.e. $\sum_{t'\in[T]} \mathbf{m}^{(i)}_{t'}=\mathbf{0}$.
	Under the setting of Theorem~\ref{thm:M error}, top $r$ eigenvalues have higher order compared with remaining eigenvalues so we have
	\begin{equation}\label{eq:M_F}
		\sum_{k=1}^r\lambda_k(\mb^{(i)})
	=\|\mm^{(i)}\|_F^2=\sum_{t=1}^T \|\mathbf{m}^{(i)}_{t}\|^2 
	\asymp T\rho_n\tilde\rho.
	\end{equation}
	Therefore when $r$ is bounded and $\mb^{(i)}$ has bounded condition number, we have $\lambda_r(\mb^{(i)})\asymp T\rho_n\tilde\rho$ and Eq.~\eqref{eq:MO-M} in Theorem~\ref{thm:M error} becomes
	$$
	\frac{1}{\sqrt{T}}\min_{\mo\in\mathcal{O}_r}\|\hat\mm^{(i)}\mo-\mm^{(i)}\|_F\lesssim
 \frac{T^{1/2}(n\rho_n)^{1/2}}{n (T\rho_n\tilde\rho)^{1/2}}
 \lesssim \frac{1}{(n\tilde\rho)^{1/2}}
	$$
	with high probability. Notice by Eq.~\eqref{eq:M_F} we have $\frac{1}{\sqrt{T}}\|\mm^{(i)}\|_F\asymp (\rho_n\tilde \rho)^{1/2}$, and therefore there exists orthogonal matrix $\mw^{(i)}_\mm$ such that $\hat\mm^{(i)}\mw^{(i)}_\mm\xrightarrow{w.h.p.} \mm^{(i)}$ {\color{black}under the assumptions $n\rho_n\tilde\rho^2=\omega(1)$ and $T=O(n^{\ell})$ for some $\ell>0$}.
\hspace*{\fill}  $\square$ 

\subsection{Proof of Proposition~\ref{prop:Dstar error2}}
By Theorem~\ref{thm:Dstar error} and the result $\lambda_r(\mb^{(i)})\asymp T\rho_n\tilde\rho$ derived in Proposition~\ref{prop:M error2}, the desired result of $\hat\md^{\star}$ follows.
\hspace*{\fill} $\square$

\subsection{Proof of Theorem~\ref{thm:M error2} and Theorem~\ref{thm:Dstar error2}}
For any $t_1,t_2\in[T]$, according to the proof of Theorem~6 in \cite{athreya2024discovering} we have
\begin{equation}\label{eq:random X}
	|(\hat\md^{(i)})^{\circ 2}-(\breve\md^{(i)})^{\circ 2}|_{\max}\lesssim n^{-1}(n\breve\rho_n)^{1/2}\log n
\end{equation}
with high probability. 
Based on Eq.~\eqref{eq:random X}, with almost identical analysis of Theorem~\ref{thm:M error} and Theorem~\ref{thm:Dstar error}, the desired results are derived.
\hspace*{\fill} $\square$

\end{document}